\documentclass[%
    aps,superscriptaddress,  
    pra, 
    nofootinbib,%
    reprint,
    a4paper,
    ]{revtex4-1} 

\usepackage{nicefrac}
\usepackage{amssymb}
\usepackage[english]{babel}
\usepackage{soul} 

\usepackage{array} 

\usepackage{enumitem}
\setlist{noitemsep} 
\setitemize{topsep=3pt}

\usepackage{hyperref}
\usepackage{xcolor}
\definecolor{mylinkcolor}{rgb}{0,0,0.4} 
\hypersetup{unicode=true, %
  bookmarksnumbered=false,bookmarksopen=false,bookmarksopenlevel=1, %
  breaklinks=true,pdfborder={0 0 0},colorlinks=true}%
\hypersetup{%
  anchorcolor=mylinkcolor,citecolor=mylinkcolor, %
  filecolor=mylinkcolor,linkcolor=mylinkcolor, %
  menucolor=mylinkcolor,runcolor=mylinkcolor, %
  urlcolor=mylinkcolor}%

\usepackage{bbm} 
\usepackage{amsmath}
\usepackage{graphicx}
\usepackage{amsthm}

\makeatletter
\def\pdfstartlink@attr{attr{/Border[0 0 0 [1 5] ]/H/I/C[0 1 1]}}%
\def\@@Doi#1{\textcolor{mylinkcolor}{#1}\@@endlink}
\makeatother

\def\ForTexCount\section#1{} 

\newtheorem{thm}{Theorem}
\newtheorem{prop}[thm]{Proposition}
\newtheorem{lemma}[thm]{Lemma}
\newtheorem{cor}[thm]{Corollary}

\newtheorem{definition}[thm]{Definition}

\makeatletter
\newcommand\bibalias[2]{%
  \@namedef{bibali@#1}{#2}%
}

\newtoks\biba@toks
\let\bibalias@oldcite\cite{}
\def\cite{%
  \@ifnextchar[{%
    \biba@cite@optarg%
  }{%
    \biba@cite{}%
  }%
}
\newcommand\biba@cite@optarg[2][]{%
  \biba@cite{[#1]}{#2}%
}
\newcommand\biba@cite[2]{%
  \biba@toks{\bibalias@oldcite#1}%
  \def\biba@comma{}%
  \def\biba@all{}%
  \@for\biba@one:=#2\do{%
    \@ifundefined{bibali@\biba@one}{%
      \edef\biba@all{\biba@all\biba@comma\biba@one}%
    }{%
      \PackageInfo{bibalias}{%
        Replacing citation `\biba@one' with `\@nameuse{bibali@\biba@one}'
      }%
      \edef\biba@all{\biba@all\biba@comma\@nameuse{bibali@\biba@one}}%
    }%
    \def\biba@comma{,}%
  }%
  \immediate\write\@auxout{\noexpand\bgroup\noexpand\renewcommand\noexpand\citation[1]{}\noexpand\citation{#2}\noexpand\egroup}%
  \edef\biba@tmp{\the\biba@toks{\biba@all}}%
  \biba@tmp%
}
\makeatother

\bibalias{Aczel1974AAP_natural}{Aczel1974}
\bibalias{Csiszar2008Ent_axiomatic}{Csiszar2008}
\bibalias{DelRio2011}{delRio2011Nature}
\bibalias{Horodecki2003PRA_NoisyOps}{Horodecki2003}
\bibalias{Lieb1999}{Lieb1999_secondlaw}
\bibalias{Lieb2013}{Lieb2013_entropy_noneq}

\definecolor{myorange}{rgb}{0.8,0.1,0}

\begin{document}

\title{Axiomatic relation between thermodynamic and information-theoretic entropies}
\author{Mirjam Weilenmann}
\affiliation{Institute for Theoretical Physics, ETH Zurich, 8093 Switzerland}
\affiliation{Department of Mathematics, University of York, Heslington, York, YO10 5DD, UK}
\author{Lea Kraemer}
\affiliation{Institute for Theoretical Physics, ETH Zurich, 8093 Switzerland}
\author{Philippe Faist}
\affiliation{Institute for Theoretical Physics, ETH Zurich, 8093 Switzerland}
\affiliation{Institute for Quantum Information and Matter, California Institute of Technology, Pasadena 91125, USA}
\author{Renato Renner}
\affiliation{Institute for Theoretical Physics, ETH Zurich, 8093 Switzerland}
\date{\today}

\begin{abstract} 
  Thermodynamic entropy, as defined by Clausius, characterises macroscopic
  observations of a system based on phenomenological quantities such as temperature and heat.
  In contrast, information-theoretic entropy, introduced by Shannon, is a measure of
  uncertainty. In this Letter, we connect these two notions of entropy, using an axiomatic framework for thermodynamics
  [Lieb, Yngvason, Proc.\@ Roy.\@ Soc.\@ (2013)].
  In particular, we obtain a direct relation between the Clausius
  entropy and the Shannon entropy, or its generalisation to quantum systems, the von
  Neumann entropy. More generally, we find that entropy measures relevant in
  non-equilibrium thermodynamics correspond to entropies used in one-shot information
  theory.
\end{abstract}

\maketitle

\ForTexCount\section{Main Text}
Entropy plays a central role both in thermodynamics and in information theory. 
This is remarkable, as the two theories appear to be fundamentally different. 
Thermodynamics is a phenomenological theory, concerned with the description of large physical systems, 
such as steam engines or fridges. 
It relies on concepts  like work or heat, which are defined in terms of macroscopic observables. 
Information theory, on the other hand, deals with ``knowledge'' on a rather abstract level. 
Like statistical mechanics, it refers to the microscopic states of  a system, 
such as the values of the individual bits stored in a memory device.  

Accordingly, the notion of entropy is rather different in the two theories. In
thermodynamics, entropy is a function of the macroscopic state of a physical system
which describes phenomenologically which processes are possible independently of any
microscopic model.
Following Clausius, it is conventionally defined in terms of the heat transfer into a
system at a given temperature, and it lends its operational significance from the second
law~\cite{Clausius1854AdP_veraenderte,Clausius1865AdP_verschiedene}. In information
theory, entropy was originally introduced by Shannon to quantify the information content
of data or, equivalently, the uncertainty one has about them~\cite{Shannon1948BSTJ}.
Operationally, it characterises properties such as the compression length, i.e.\ the
minimum number of bits needed to store the data. Mathematically, the Shannon entropy is a
function of the probability distribution of the random variable that models the data. The
von Neumann entropy~\cite{vNeumann1932} provides a generalisation of this concept to the case where information is represented by the state of a quantum-mechanical system.

The information theoretic entropy is formally equivalent to the entropy function from statistical mechanics. This relation is conceptually justified through
Landauer's
principle~\cite{Landauer1961_5392446Erasure,Bennett1982IJTP_ThermodynOfComp,BookLeffRex2010}.
It entails that the loss of information
in an erasure operation on a system, and hence the decrease of its information-theoretic entropy, is paired with a heat dissipation in the system's environment.
Arguments for Landauer's principle 
start from microscopic considerations, 
for 
example using standard tools from statistical
mechanics~\cite{Piechocinska2000PRA,Shizume1995_HeatGeneration,Reeb2014NJP_improved}, or
explicit microscopic models inspired from information theory~\cite{%
  Janzing2000_cost,%
  Dahlsten2011NJP_inadequacy,%
  DelRio2011,%
  Brandao2013_resource,%
  Aberg2013_worklike,%
  Horodecki2013_ThermoMaj,%
  Egloff2012,%
  Faist2015,%
  Skrzypczyk2014NComm_individual,%
  Brandao2015,%
  Halpern2014%
}.
An alternative view on the conceptual connection of these entropy measures was proposed by Jaynes in terms of his maximum entropy principle~\cite{Jaynes1957}.
See also~\protect\cite{Maxwell1871,%
Tribus1971,%
Lindblad1973,%
Wherl1978,%
Bennett1982IJTP_ThermodynOfComp} for a selection of related approaches.

In this Letter, we take a different approach: we show that the information-theoretic entropy results from applying the definition of thermodynamic entropy 
to quantum resource theories. This connection is not based on a model borrowed from one particular theory, instead it relates the theory of information to
that of phenomenological thermodynamics on an
axiomatic level.

Our approach relies on a framework by Lieb and
Yngvason~\cite{Lieb1998,Lieb1999,Lieb2013,Lieb2014}, who give a derivation of
thermodynamics, and in particular of the thermodynamic entropy, based on abstract axioms.
These axioms identify basic properties of a thermodynamic system, which we find to be also fulfilled in the context of resource theories.  As
a consequence, there is also in this context a state function analogous to the
thermodynamic entropy. As we show, the state function in question is none else than the
information-theoretic entropy itself. This provides a novel connection between the
thermodynamic and information-theoretic entropies.

In particular, this connection extends to the min- and max-entropy, ``single-shot'' generalizations of the von Neumann entropy.
These have been introduced to characterize single instances of information-theoretic tasks~\cite{PhdRenner2005_SQKD,Koenig2009IEEE_OpMeaning}.
We show that they, too, can be obtained from the same axiomatic approach as thermodynamic entropy.  In order to
demonstrate this, we consider an extension of Lieb and Yngvason's
framework to nonequilibrium states~\cite{Lieb2013_entropy_noneq},
with which these entropy measures are recovered.
Our work bears some
resemblance to the study of entanglement theory using similar
tools~\cite{Vedral2002PRL_uniqueness,Brandao2008_secondlaw,Brandao2010CMP_reversible}.

The remainder of this Letter is organised as follows. We start with a summary of the
Lieb-Yngvason framework for thermodynamics. 
As a first technical contribution, we show that the framework is applicable to
a microscopic description of thermodynamic systems by resorting to a quantum resource theory.
We then show that the corresponding entropy
measures defined within the framework coincide with information-theoretic entropies
(Proposition~\ref{prop:entropies}). We subsequently extend these considerations to
other thermodynamic quantities such as the Helmholtz free energy  (Proposition~\ref{prop:reservoir}).

\textit{Lieb and Yngvason's Approach.}---The axiomatic framework by Lieb and Yngvason
follows a history of
developments towards a
mathematically rigorous treatment of
thermodynamics~\cite{%
  Caratheodory1909,%
  Giles1964,%
  Lieb1998,%
  Lieb1999,%
  Lieb2013,%
  gyftopoulos1991thermodynamics,%
  Gyftopoulos2005,%
  Zanchini2011_thermodynamics,%
  Zanchini2014,%
  Hatsopoulos1,%
  Hatsopoulos2a,%
  Hatsopoulos2b,%
  Hatsopoulos3%
}.
Lieb and Yngvason~\cite{Lieb1998,Lieb1999,Lieb2013}
consider the set $\Gamma$ of all equilibrium states of a thermodynamic system and equip
this space with an order relation, denoted by $\prec$.  For $X$ and $Y \in \Gamma$, $X \prec
Y$ means that the state $Y$ is ``adiabatically accessible'' from the state $X$
``by means of an interaction with some device consisting of some auxiliary system and a weight in such a way that the auxiliary system returns to its initial state at the end of
the process, whereas the weight may have risen or fallen''~\cite{Lieb1998}. It is moreover assumed that two systems $X$, $X'$ can be composed, denoted by $(X,X')$, as well as that a system $X$ can be scaled by a factor $\lambda$, denoted by $\lambda X$. The scaling corresponds to considering a fraction $\lambda$ of the substance in $X$.

Provided $\prec$ obeys some natural axioms, Lieb and Yngvason show that there is an
essentially unique \emph{thermodynamic entropy} $S$ that correctly characterizes all
possible state transformations, which is given by
\begin{subequations}
  \label{eq:entropy-both}
  \begin{alignat}{2}
    S(X) &= \sup &&\left\{ \lambda : ((1-\lambda) X_{\mathrm{0}}, \lambda X_{\mathrm{1}}) \prec X \right\} \label{eq:entropy} \\
    &= \inf &&\left\{ \lambda : X \prec ((1-\lambda) X_{\mathrm{0}}, \lambda
      X_{\mathrm{1}})\right\}\ , \label{eq:entropy2}
  \end{alignat}
\end{subequations}
where $X_0$ and $X_1$ are two reference states whose choice alter $S$ only by an affine
change of scale~\footnote{This refers to a transformation of the form $c_1 \cdot S + c_0$ with $c_1 >0$.}.
Intuitively, if the state $X$ can be reached adiabatically from $X_0$, and $X_1$ can be attained from
$X$, then the entropy $S(X)$ is defined as the optimal $\lambda$ such that
the state $X$ can be created from $X_{\mathrm{0}}$ and $X_{\mathrm{1}}$ combined at a ratio $(1-\lambda):\lambda$ 
by an adiabatic process. Physically,  $S$ corresponds to the usual thermodynamic entropy as defined by Clausius via the heat $\delta Q_\mathrm{rev}$ transferred into a system at a given temperature $T$ in a reversible process,
$dS=\delta Q_\mathrm{rev}/T$. 

Lieb and Yngvason have extended this framework to include certain non-equilibrium states~\cite{Lieb2013}. 
The states of the corresponding extended state space, $\Gamma_{\mathrm{ext}}$, obey
weaker axioms than those of $\Gamma$. For instance, they may not be scalable.  
The entropy $S$ can in general not be uniquely extended to $\Gamma_{\mathrm{ext}}$. 
However, one can bound all monotonic extensions of $S$ to
$\Gamma_{\mathrm{ext}}$ from below and above by two quantities $S_-$ and $S_+$~\cite{Lieb2013}.
These quantities give necessary criteria as well as sufficient criteria for adiabatic transitions between thermodynamic non-equilibrium states. Here we use instead slightly adapted quantities, defined as
\begin{subequations}
  \begin{alignat}{2}
    \tilde{S}_{\mathrm{-}}(X) &= \sup &&\left\{ \lambda : ((1-\lambda) X_{\mathrm{0}},
      \lambda X_{\mathrm{1}}) \prec X \right\}
    \label{eq:S-tilde-minus}\\
    \tilde{S}_{\mathrm{+}}(X) &= \inf &&\left\{ \lambda : X \prec ((1-\lambda)
      X_{\mathrm{0}}, \lambda X_{\mathrm{1}})\right\} \ .
    \label{eq:S-tilde-plus}
  \end{alignat}
\end{subequations}
While they are essentially equivalent to $S_-$ and $S_+$,
they are the more natural quantities within the
context we consider 
(cf.\ also Appendix for a detailed analysis).

\textit{Information-Theoretic Entropy Measures.}---%
Information theory is concerned with data and their processing. In quantum information
theory, which we consider here for generality, data is encoded in quantum systems~\footnote{These
include classical systems as a special case.}, whose states are described by the
density operator formalism.
Throughout this Letter, we restrict our attention to finite-dimensional quantum systems.
Information is quantified using an information-theoretic entropy measure, most commonly the
von Neumann entropy $H(\rho)=-\operatorname{tr} \left( \rho \log \rho \right)$. 
Note that $\log$ denotes the logarithm with respect to base $2$ in this Letter. 
The information-theoretic significance of the von Neumann entropy can be established in various ways, e.g., axiomatically~\cite{%
  Shannon1948BSTJ,%
  Aczel1974AAP_natural,%
  Ochs1975_axiomatic,%
  Csiszar2008Ent_axiomatic,%
  Baumgartner2014FP_characterizing%
}.
Other useful entropy measures are the min- and the max-entropy. The min-entropy is defined as
$
H_{\mathrm{min}}(\rho) = - \log \|\rho\|_{\infty}
$,
where $\|\rho\|_{\infty}$ denotes the maximal eigenvalue of $\rho$. 
Operationally, it describes the amount of randomness that can be extracted deterministically 
from data in state $\rho$~\cite{Tomamichel2010IEEE_Duality,Koenig2009IEEE_OpMeaning}.
The max-entropy is defined as 
$
H_{\mathrm{max}}(\rho) = \log \operatorname{rank}\rho
$,
and quantifies the number of qubits needed to store data in state~$\rho$~\cite{PhdRenner2005_SQKD}. 

\textit{Equilibrium States and Order Relation in the Microscopic Picture.}---%
To apply Lieb and Yngvason's framework to a microscopic description of
systems, as employed in information theory, we need to formally specify the various ingredients (such
as the order relation) that the abstract framework requires.
First, we identify the set of ``equilibrium states'' of an information-bearing quantum system.
These are defined as the class of states represented by flat density operators, that is, operators whose non-zero eigenvalues are all
equal~\cite{Horodecki2013_ThermoMaj}. They stand out due to their scalability and comparability, like the equilibrium states in the thermodynamic framework.
(The term ``equilibrium'' follows Lieb and Yngvason's terminology.)

We define the composition of states as their tensor product. This does not exclude the possibility of correlations being established between subsystems but it merely asserts that before any interaction takes place these systems are independent~\cite{Horodecki2013_ThermoMaj,Aberg2013_worklike}.
For the composition with an ancilla, this can be ensured by choosing an ancilla that has never interacted with the system before or that has been decoupled through other interactions.

The notion of an ``adiabatic process'', in the
sense of Lieb and Yngvason,
translates to any combination of the following three quantum operations; the order relation characterises which states can be transformed to which others via such a process: 
\begin{itemize}
  \item composition with an extra ancilla system in an equilibrium state; 
  \item  reversible and energy-preserving
interaction of the system and the
ancilla with a weight system~\cite{Aberg2014PRL_catalytic};
  \item removal of the ancilla system, whose final reduced state must be the same as its initial state.
\end{itemize}
Note that these
operations are independent of the
particular model used to describe the
weight system. In fact, the weight
system should be understood as a
representative for any work storage
system, which may also be modelled
via a potential, as e.g.\ in~\cite{Allahverdyan2004}.
It can also
be shown that the above operations are
equivalent to the so-called \emph{noisy
  operations}~\cite{Horodecki2003, Gour2013arXiv_resource}, a quantum resource theory that plays a prominent role in information theory. In general, a resource theory is defined by restricting the set of all quantum operations to a subset, the \emph{allowed transformations}.
Given such a set of allowed
transformations, the aim of a
resource-theoretic analysis is now to
characterise which states can be
interconverted and which tasks can be achieved with these transformations. In the case of noisy operations, a state transformation from $\rho$ to
$\sigma$ 
is possible if and only if
$\rho\prec_M\sigma$~\cite{Horodecki2003, Gour2013arXiv_resource}, where $\prec_M$ denotes the matrix majorization
relation~\cite{BookBhatiaMatrixAnalysis1997,Karamata}.
(See also Appendix~\ref{sec:resources} for further details on quantum resource theories.)

We define scaling a quantum system by $\lambda \in \mathbbm{N}$ 
to mean combining $\lambda$ such systems, which coincides with the composition operation. For $\lambda=2$, for instance, the state $\rho$ is scaled to a state $\rho\otimes\rho$.~\footnote{Note that correlations within the system do not affect this scaling operation. The latter corresponds to a scaling from an outside perspective: without splitting the system up and probing it, the knowledge about its internal properties, such as the entanglement of individual particles, is not accessible. Our operations thus avoid the potential issues with entanglement in microscopic systems that were raised in~\cite{Lieb2014}.}.
We then extend this scaling operation to any $\lambda \in \mathbbm{R}_{\geq 0}$, as is explained in Appendix~\ref{sec:results}.

\textit{Main Results.}---%
We apply Lieb and Yngvason's framework to the microscopic model
detailed above.
We show that Lieb and Yngvason's axioms are fulfilled in these settings,
and hence the implications of the framework apply. This allows us to derive corresponding entropy measures, which are furthermore unique for equilibrium states. 
Specifically, we obtain the following statement, the formal derivation of which is based on properties of the matrix majorization relation. The full proofs may be found in Appendix~\ref{sec:results}.
\begin{prop} \label{prop:entropies}%
Let states be ordered by the relation $\prec_{\mathrm{M}}$ as described above, and let the equilibrium states be those with flat spectrum.
Then the unique thermodynamic entropy function $S$ for equilibrium states coincides with 
the von Neumann entropy $H$.
Furthermore, the two entropic quantities relevant for non-equilibrium states, $\tilde{S}_{\mathrm{-}}$ and $\tilde{S}_{\mathrm{+}}$, equal $H_{\mathrm{min}}$ and $H_{\mathrm{max}}$ respectively.
\end{prop}
In other words, quantum information theory can be viewed as an instance of
  thermodynamics in the sense of Lieb and Yngvason, and the corresponding thermodynamic
  entropy is precisely the information-theoretic entropy.

It is natural to consider other, non-adiabatic processes within the same mathematical framework, for instance scenarios where the system is in contact with additional reservoirs.
In case the system interacts with a heat bath, the equilibrium states can be taken to be the thermal states of fixed temperature $T$, corresponding to the temperature of the bath, as also known from the canonical ensemble in statistical mechanics.

In such a setting, the processes on the system of interest are the \emph{thermal operations}, introduced in~\cite{Janzing2000_cost,Brandao2013_resource,Horodecki2013_ThermoMaj}. These consist of
\begin{itemize}
	\item composition with an ancillary system in a thermal state relative to the heat bath at temperature $T$; 
	\item reversible and energy-conserving transformation of the system and the ancilla; 
	\item removal of any subsystem.
\end{itemize}
These operations have been extensively studied and used to understand and characterize
possible thermodynamic operations in information-theoretic terms (see for instance~\cite{Brandao2013_resource,Horodecki2013_ThermoMaj,Brandao2015}).  As in~\cite{Horodecki2013_ThermoMaj,Brandao2015}, we restrict our analysis to states of the
system that are block diagonal in the energy eigenbasis.  This restriction allows us to
avoid technical difficulties when dealing with coherent superpositions of energy levels;
other works have also studied general transformations beyond this
restriction~\cite{Korzekwa2015,Lostaglio2015arXiv_noncommutativity,Guryanova2016NatComm_multiple}.

The thermal operations are characterized with the mathematical notion of
thermo-majorization~\cite{Horodecki2013_ThermoMaj}, denoted here by $\prec_\mathrm{T}$. 
This order relation obeys Lieb and Yngvason's axioms, and, as before, we may deduce
corresponding thermodynamic entropy measures.

\begin{prop}\label{prop:reservoir} For states ordered by means of thermal operations through the relation $\prec_{\mathrm{T}}$ the unique function 
$S$ for thermal states coincides with 
the Helmholtz free energy $F$, and the two quantities $\tilde{S}_{\mathrm{-}}$ and $\tilde{S}_{\mathrm{+}}$, relevant for non-thermal states, correspond to 
$F_{\mathrm{max}}$ and $F_{\mathrm{min}}$ from~\cite{Horodecki2013_ThermoMaj,Brandao2015} respectively.
\end{prop}
The single-shot measures $F_\mathrm{min}$ and
$F_\mathrm{max}$ were introduced in~\cite{Horodecki2013_ThermoMaj,Brandao2015} to
describe the work needed for the formation of a state as well as to characterise the extractable work.
Proposition~2 follows analogously to Proposition~1, but with the thermo-majorization relation instead of the (standard) majorization; the former includes
an additional transformation known as Gibbs-rescaling~\cite{Ruch1975_diagramlattice,%
Ruch1976_increasingmixing,%
Mead1977JCP_mixing,%
Egloff2012,%
Horodecki2013_ThermoMaj} (cf.\ Appendix~\ref{sec:results2}).

Scenarios defined relative to other types of reservoirs, such as a
particle~\cite{Halpern2014} or an angular momentum reservoir~\cite{Barnett2013_beyond,Vaccaro2011}, yield analogous results. Various settings,
along with their corresponding order relation, equilibrium states, and resulting state
functions, are summarized in Table~\ref{table:questions}.
\begin{table}\scriptsize
\begingroup
\def\inlinedisplayeqn#1{\vspace*{0.5ex}$\displaystyle #1$\vspace*{0.5ex}}
\begin{tabular}{|m{1.4cm}|m{1.5cm}|m{2.7cm}|m{1.2cm}|m{1.1cm}|}
    \hline
    \hspace{0pt}\textbf{Setting}  &\hspace{0pt}\textbf{Processes} &\hspace{0pt}\textbf{Equilibrium States} &\hspace{0pt}\textbf{$S$} & \hspace{0pt}$\tilde{S}_{\mathrm{-}}$, $\tilde{S}_{\mathrm{+}}$\\ 
		\hline
    \hspace{0pt}Isolated system & \hspace{0pt}Noisy operations, $\prec_{\mathrm{M}}$ &
 \inlinedisplayeqn{\sum_{i}\frac{1}{\operatorname{rank}\rho}\left| X_{i}  \right\rangle \left\langle X_{i} \right|} & \hspace{0pt}Entropy $H$ & \hspace{0pt}$H_{\mathrm{\mathrm{min}}}$, $H_{\mathrm{max}}$\\ 
		\hline
    \hspace{0pt}\raggedright Interaction with a heat bath & \hspace{0pt}Thermal operations, $\prec_{\mathrm{T}}$ & \inlinedisplayeqn{\sum_{i} \frac{ e^{-\beta E_{i}}}{Z} \left| E_{i}  \right\rangle \left\langle E_{i} \right|} & \hspace{0pt}Free Energy $F$ & \hspace{0pt}$F_{\mathrm{min}}$, $F_{\mathrm{max}}$\\ 
		\hline
    \hspace{0pt}\raggedright Interaction with a heat bath and a particle reservoir &
    \hspace{0pt}N,T-operations, $\prec_{\mathrm{N,T}}$ & 
    \inlinedisplayeqn{\sum_{i}\frac{e^{-\beta(E_{i}-N_{i}\mu)}}{\mathcal{Z}}\newline
      \hspace*{1em}\times\left| E_{i},N_{i} \right\rangle\left\langle E_{i},N_{i} \right|}
    & \hspace{0pt}Grand potential $\Omega$ & \hspace{0pt}$\Omega_{\mathrm{min}}$, $\Omega_{\mathrm{max}}$\\ 
	  \hline
    \hspace{0pt}\raggedright Interaction with an angular momentum reservoir & \hspace{0pt}J-operations, $\prec_{\mathrm{J}}$ & \inlinedisplayeqn{\sum_{i}\frac{e^{ - \hbar \gamma J_{i}}}{\mathcal{Z_{\mathrm{J}}}}\left| J_{i} \right\rangle \left\langle J_{i} \right|}  & \hspace{0pt}Potential $S_{\mathrm{J}}$ & \hspace{0pt}$\tilde{S}_{\mathrm{J-}}$, $\tilde{S}_{\mathrm{J+}}$\\
          \hline
    \end{tabular}
\endgroup
\caption{An overview on the application of the Lieb-Yngvason framework to various scenarios, as specified in the first column of the table.
The first two rows summarize the results of Propositions~1 and 2 (cf.\ Appendix~\ref{sec:results3} for details on the other two scenarios).}
\label{table:questions}
\end{table}

\textit{Discussion.}---%
\begin{figure}%
\centering
\includegraphics[width=\columnwidth]{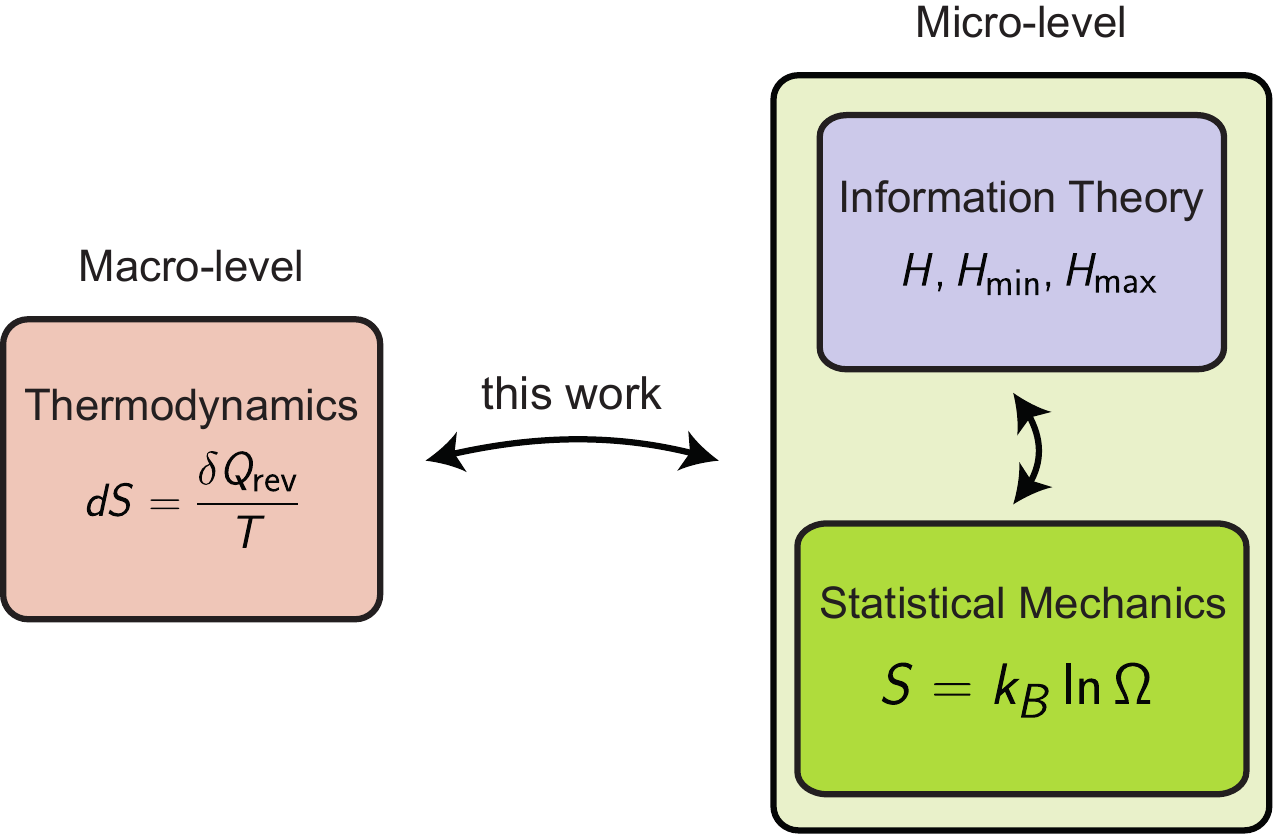}
\caption{The contribution of this Letter is to
draw a connection between thermodynamic and information-theoretic entropy based on formal axioms and general principles such as additivity and monotonicity that are satisfied within both theories (Proposition~\ref{prop:entropies}).
This replaces the textbook identification of entropy in thermodynamics and statistical mechanics based merely on analogies, such as common behaviour of the two entropy functions in specific models like the ideal gas. 
The entropy functions in information-theory and statistical mechanics are already formally equivalent,  and conceptual links have been established through Jaynes' maximum entropy principle~\protect\cite{Jaynes1957} as well as through the works of 
Landauer~\protect\cite{Landauer1961_5392446Erasure} and Bennett~\protect\cite{Bennett1982IJTP_ThermodynOfComp} (based on explicit protocols that convert one into the other in an underlying microscopic model).
}
\label{fig:entropymap}
\end{figure}%
We have shown that, with minor adaptations, Lieb and Yngvason's framework is directly 
applicable to quantum resource theories, allowing us to put thermodynamic and information-theoretic entropy on the same footing (see also Figure~\ref{fig:entropymap}). 

More generally, our approach points out the formal and conceptual parallels of phenomenological thermodynamics and quantum resource theories.
It thus relates the classical thermodynamic description of a system (based on macroscopic properties) to a description in terms of the information (about the microscopic degrees of freedom) held by an observer.  
This underscores that the description of a physical system, in particular its characterisation with an entropy function, is observer-dependent, hence, subjective.

We have justified the use of the majorization relation based on Lieb and Yngvason's adiabatic
operations. This relation also occurs naturally in information theory, however. Indeed, it
expresses for instance an encoding operation or the inverse of a randomness extraction
process~\footnote{By ``inverse'' we mean the map that inverts the extractor
  function: It is a stochastic process and thus describes a noisy operation.}.
We have also shown that order relations characterising other resource theories are compatible with the axiomatic framework and allow us to derive corresponding ``entropy functions'' (cf.\ Proposition~2 and Appendix). We have in this way recovered the expression of the Helmholtz free energy as well as corresponding single-shot counterparts~\cite{Horodecki2013_ThermoMaj,Brandao2015}.

For a system interacting with a reservoir, our approach is so far limited to
states that are block diagonal in a corresponding eigenbasis, e.g.\@ in the particular
case of a heat bath in the energy eigenbasis~\cite{Horodecki2013_ThermoMaj}. We leave for future work the question whether our results also hold for states that do not satisfy this property,
which have been studied in~\cite{Lostaglio2015,Korzekwa2015,Lostaglio2015arXiv_noncommutativity,Guryanova2016NatComm_multiple,Winter2015,Baumgratz2014}.

We expect that the approach presented in this work can be applied to relate other thermodynamic and information-theoretic quantities. 
For example, by slightly changing the order
relation to a ``smooth majorization relation'', we presume it to yield so called
smooth entropy measures~\cite{PhdRenner2005_SQKD,Renner2005}.  
Furthermore, we have not treated processes where quantum side information about the system is exploited. 
This could be useful for performing thermodynamic
operations~\cite{DelRio2011}.
We might anticipate that an appropriate extension of the Lieb-Yngvason framework would provide an axiomatic and operationally well-justified definition of the conditional entropy.

\bigskip

\textit{Acknowledgements.}---%
This project has been
supported by the Swiss National Science Foundation (SNSF) via project grant No.~200021\_153296 and
via the NCCR ``QSIT'',
by the European Research Council (ERC) via project
No.~258932, and by the COST Action MP1209.

\bigskip

\appendix 

\section{Lieb and Yngvason's axiomatic approach} \label{sec:ly}
Lieb and Yngvason have devised an axiomatic approach to derive an entropy function for 
thermodynamic equilibrium states~\cite{Lieb1998,Lieb1999_secondlaw}.
Recently, they have extended their approach to a special class of non-equilibrium states~\cite{Lieb2013_entropy_noneq}, 
which also enables them to make predictions relevant for non-equilibrium thermodynamics. We present a short summary of their approach here, focusing on the details relevant for our application. For further information we refer to the original framework~\cite{Lieb1998,Lieb1999_secondlaw,Lieb2013_entropy_noneq,Lieb2014}.

Lieb and Yngvason consider a preorder $\prec$ on a set $\Gamma$;
physically $\Gamma$ is the space of all equilibrium states of a thermodynamic system. The ordering of the states expresses the existence of adiabatic processes that convert one state into another: for two states $X,~Y \in \Gamma$ the ordering $X \prec Y$ means that there exists an adiabatic process transforming the system in state $X$ to the state $Y$. Note that an adiabatic process is explicitly defined as a process that leaves no trace on the environment except that a weight may have changed its relative position. Such transformations are also known as ``work processes''~\cite{Gyftopoulos2005}.
Mathematically, a preorder is reflexive and transitive, but in contrast to a partial order not antisymmetric;
two elements $X$ and $Y \in \Gamma$ that satisfy $X \prec Y$ as well as 
$Y \prec X$ are not necessarily the same element of the set $\Gamma$.
Whenever both relations $X \prec Y$ and $Y \prec X$ hold, this is denoted by $Y \sim X$,
while $X \prec Y$ but not $Y \prec X$ is denoted $X \prec \prec Y$.

Elements $X_{\mathrm{1}} \in \Gamma_{\mathrm{1}}$ and $X_{\mathrm{2}} \in \Gamma_{\mathrm{2}}$
of possibly different $\Gamma_{\mathrm{1}}$ and $\Gamma_{\mathrm{2}}$ can be composed, 
physically meaning that they can be considered as one composite system before any interaction takes place. 
The composed state is formally denoted as 
$(X_{\mathrm{1}},~X_{\mathrm{2}}) \in \Gamma_{\mathrm{1}} \times \Gamma_{\mathrm{2}}$.
Note that the Cartesian product $\Gamma_{\mathrm{1}} \times \Gamma_{\mathrm{2}}$ 
stands for the space of all states
$(X_{\mathrm{1}},~X_{\mathrm{2}})$ of the composed system, where the composition operation is associative and commutative.
Furthermore, any element $X \in \Gamma$ can be scaled, i.e., for any $\lambda \in \mathbbm{R}_{\geq 0}$ one can define a scaled element denoted as $\lambda X \in \lambda \Gamma$.\footnote{ 
The scaling is required to obey $1 X = X$ as well as $\lambda_{\mathrm{1}}(\lambda_{\mathrm{2}} X) = (\lambda_{\mathrm{1}}\lambda_{\mathrm{2}}) X$.
For the sets $\Gamma$, the required properties are $1 \Gamma = \Gamma$ and $\lambda_{\mathrm{1}}(\lambda_{\mathrm{2}} \Gamma) = (\lambda_{\mathrm{1}}\lambda_{\mathrm{2}}) \Gamma$,
where $\lambda \Gamma$ denotes the space of scaled elements $\lambda X$.}
Scaling a system by a factor $\lambda$ means taking $\lambda$ times the amount of substance contained in the original system.

The order relation $\prec$ satisfies by assumption the following six axioms E1 to E6 as well as the Comparison Hypothesis.
\begin{itemize}
\item Reflexivity (E1): $X \sim X$.
\item Transitivity (E2): $X \prec Y$ and $Y \prec Z$ $\Rightarrow$ $X \prec Z$.
\item Consistent composition (E3): $X \prec Y$ and $X' \prec Y'$ $\Rightarrow$ $(X,~X') \prec (Y,~Y')$.
\item Scaling invariance (E4): $X \prec Y$ $\Rightarrow$ $\lambda X \prec \lambda Y$, $\forall \ \lambda > 0$. 
\item Splitting and recombination (E5): For $0<\lambda<1$, $X \sim (\lambda X,~(1-\lambda) X)$.
\item Stability (E6): If $(X,~\varepsilon Z_{\mathrm{0}}) \prec (Y,~\varepsilon Z_{\mathrm{1}})$ for a sequence of scaling factors $\varepsilon \in \mathbbm{R}$ tending to zero, 
then $X \prec Y$.
\item Comparison Hypothesis: Any two elements in a set $(1-\lambda)\Gamma \times \lambda \Gamma$ with $0 \leq \lambda \leq 1$ are related by $\prec$~\footnote{Note that Lieb and Yngvason do not adopt the Comparison Hypothesis as an axiom but rather derive it from additional axioms about thermodynamic systems~\cite{Lieb1998,Lieb1999_secondlaw}.}.
\end{itemize}

The above axioms also directly imply the so called cancellation law~\cite{Lieb1998,Lieb1999_secondlaw}: Let $X$, $Y$, and $Z$ be any states, then
\begin{equation}\label{eq:cancellation}
\left(X,~Z \right) \prec \left(Y,~Z \right) \ \Rightarrow  \ X \prec Y.
\end{equation}
Systems in equilibrium can thus not be used to enable state transformations between the states $X$ and $Y$.

Lieb and Yngvason's contribution concerns possible ``entropy functions''. More precisely, they show that there is a (essentially) unique 
real valued function $S$ on the space of all equilibrium states of thermodynamic systems that satisfies the following:
\begin{itemize}
\item Additivity: For any two states $X \in \Gamma$ and $X' \in \Gamma'$, $S((X,X'))=S(X)+S(X')$.
\item Extensivity: For any $\lambda > 0$ and any $X \in \Gamma$, $S(\lambda X)=\lambda S(X)$.
\item Monotonicity: For two states $X$ and $Y$ $\in \Gamma$, $X \prec Y \Leftrightarrow S(X) \leq S(Y)$.
\end{itemize}
Lieb and Yngvason's second law is rephrased in the following theorem.
\begin{thm}[Lieb \& Yngvason]
  \label{thm:entropythm}
  Provided that the six axioms E1 to E6 as well as the Comparison Hypothesis are fulfilled, 
  there exists a function $S$ that is additive under composition, extensive in the scaling and monotonic with respect to $\prec$. 
  Furthermore, this function $S$ is unique up to an affine change of scale $C_1 \cdot S + C_0$ with $C_1 > 0$. 
\end{thm}
For a state $X \in \Gamma$, the unique function $S$ is given as
\begin{alignat}{2} S(X) &= \sup &&\left\{ \lambda : ((1-\lambda) X_{\mathrm{0}}, \lambda X_{\mathrm{1}}) \prec X \right\} \label{eq:entropyoriginal1}\\
&= \inf  &&\left\{ \lambda : X \prec ((1-\lambda) X_{\mathrm{0}}, \lambda X_{\mathrm{1}}) \right\} \label{eq:entropyoriginal2}
, \end{alignat}
where the elements $X_{\mathrm{0}} \prec \prec X_{\mathrm{1}} \in \Gamma$ may be chosen freely, and
change the function $S(X)$ only by an affine transformation as stated in the theorem.
In this sense, the states $X_\mathrm{0},X_\mathrm 1$ define a gauge.

The expressions~\eqref{eq:entropyoriginal1} and~\eqref{eq:entropyoriginal2} assume that
$X_\mathrm 0 \prec X \prec X_\mathrm 1$, and yield a value $S(X)$ which is between zero
and one.  However it is straightforward to extend this definition to states $X$ with
$X\prec X_\mathrm 0$ or $X_\mathrm1 \prec X$.  Indeed, following the idea in
Ref.~\cite{Lieb1999_secondlaw}, the expression
\begin{align}
  ((1-\lambda) X_\mathrm 0, \lambda X_\mathrm 1) \prec X
  \label{eq:expression-in-entropy-X0X1precX}
\end{align}
is equivalent for any $\lambda'\geqslant 0$ to
$(((1-\lambda) X_\mathrm 0, \lambda X_\mathrm 1), \lambda' X_\mathrm 1) \prec (X, \lambda'
X_1)$
and hence
$((1-\lambda) X_\mathrm 0, (\lambda'+\lambda) X_\mathrm 1) \prec (X, \lambda' X_1)$.  This
allows us to consider negative $\lambda$ (while still using only positive coefficients as
scaling factors). If $\lambda<0$, by choosing $\lambda'=-\lambda$, the
condition~\eqref{eq:expression-in-entropy-X0X1precX} is to be understood as
$((1-\lambda) X_\mathrm 0) \prec (X, |\lambda| X_1)$.  If $\lambda>1$, following a similar
argument, the condition $((1-\lambda) X_\mathrm 0, \lambda X_\mathrm 1) \prec X$ should be
understood as $\lambda X_\mathrm 1 \prec ((\lambda-1) X_\mathrm 0, X)$.  The same
reasoning holds for the expression in~\eqref{eq:entropyoriginal2}.

Since the scaling is continuous, the supremum in~\eqref{eq:entropyoriginal1} and the infimum in~\eqref{eq:entropyoriginal2} are attained, and the function $S$ can be conveniently expressed as 
\begin{equation} S(X)= \left\{ \lambda : ((1-\lambda) X_{\mathrm{0}}, \lambda X_{\mathrm{1}}) \sim X \right\}. \label{eq:entropy-2} \end{equation}

The framework described above for equilibrium states can be extended to include non-equilibrium states~\cite{Lieb2013_entropy_noneq}.
To describe non-equilibrium states of a thermodynamic system, additional elements that do not have to satisfy the scaling property are introduced.
By adding such non-scalable elements to the set $\Gamma$ one obtains an extended set $\Gamma_{\mathrm{ext}}$ 
for which the following is required:
\begin{itemize}
\item N1: For any $X' \in \Gamma_{\mathrm{ext}}$ there exist $X_{\mathrm{0}}$ and $X_{\mathrm{1}} \in \Gamma$ with $X_{\mathrm{0}} \prec X' \prec X_{\mathrm{1}}$.
\item N2: Axioms E1, E2, E3, and E6, where $Z_{\mathrm{0}}$ and $Z_{\mathrm{1}} \in \Gamma$ in axiom E6, hold on $\Gamma_{\mathrm{ext}}$.
\end{itemize}
The first requirement ensures that the non-scalable elements are comparable to at least two elements $X_{\mathrm{0}}$ and $X_{\mathrm{1}}$ 
of the set $\Gamma$. Furthermore, it assures that the non-scalable elements are not at the boundary of the extended 
set $\Gamma_{\mathrm{ext}}$ with respect to the preorder $\prec$. 
For such non-equilibrium states, the following holds~\cite{Lieb2013_entropy_noneq}.
\begin{prop}[Lieb \& Yngvason] \label{prop:nonequ}
On condition that N1 and N2 hold for any non-equilibrium state $X \in \Gamma_{\mathrm{ext}}$, the two functions $S_{\mathrm{-}}$ and $S_{\mathrm{+}}$ defined as
\begin{alignat}{2} S_{\mathrm{-}}(X) &= \sup &&\left\{ S(X'): X' \in  \Gamma, X' \prec X \right\}, \label{eq:lymin} \\
S_{\mathrm{+}}(X) &= \inf &&\left\{ S(X''): X'' \in \Gamma, X \prec X'' \right\} \label{eq:lymax} \end{alignat}
bound all possible extensions $S_{\mathrm{ext}}$ of $S$ to the set $\Gamma_{\mathrm{ext}}$ that are monotonic with respect to $\prec$. 
\end{prop}
This implies that for any state $X \in \Gamma_{\mathrm{ext}}$, the attained value $S_{\mathrm{ext}}(X)$ of such an extension lies between the values $S(X')$ and $S(X'')$ of its neighboring scalable states according to the order relation $\prec$:
\begin{equation*}
S_{\mathrm{-}}(X) \leq S_{\mathrm{ext}}(X) \leq S_{\mathrm{+}}(X).
\end{equation*} 

We will prefer to work with the following alternative quantities instead, which only rely on the state $X$ and not on any neighboring equilibrium states $X'$ and $X''$:
\begin{alignat}{2} \tilde{S}_{\mathrm{-}}(X) &= \sup &&\left\{ \lambda : ((1-\lambda) X_{\mathrm{0}}, \lambda X_{\mathrm{1}}) \prec X \right\}, \label{eq:condentr1}\\
\tilde{S}_{\mathrm{+}}(X) &= \inf &&\left\{ \lambda : X \prec ((1-\lambda) X_{\mathrm{0}}, \lambda X_{\mathrm{1}}) \right\} \label{eq:condentr2}
.\end{alignat}
Operationally, $\tilde{S}_{\mathrm{-}}$ specifies the portion of the system that can maximally be in state $X_{\mathrm{1}}$ if one wants to 
create the state $X$ by composing subsystems in states $X_{\mathrm{0}}$ and $X_{\mathrm{1}}$.
The minimal portion of $X_{\mathrm{1}}$ that can be obtained by transforming $X$ into a composition of two smaller systems in states $X_{\mathrm{0}}$ and $X_{\mathrm{1}}$ is characterized by $\tilde{S}_{\mathrm{+}}$.~\footnote{Note that if a state $X$ does not obey $X_0 \prec X \prec X_1$ the definitions are consistently extended like \eqref{eq:entropyoriginal1} and \eqref{eq:entropyoriginal2}.} These new entropic bounds are thus not just characterising a position of the non-equilibrium states within the ordered set of equilibrium states but they have an operational significance. The non-equilibrium states are considered based on the same processes as the equilibrium states, namely with monotonic quantities defined analogous to \eqref{eq:entropyoriginal1} and \eqref{eq:entropyoriginal2}.
Note that in thermodynamics the two sets of bounding quantities 
$\{S_{\mathrm{-}}, S_{\mathrm{+}}\}$ and $\{ \tilde{S}_{\mathrm{-}},\tilde{S}_{\mathrm{+}} \}$ 
coincide, as due to the continuity of the thermodynamic quantities an equilibrium state $X' \in \Gamma$ with 
$X' \sim ((1-\lambda) X_{\mathrm{0}}, \lambda X_{\mathrm{1}})$
exists for any $\lambda$. The two sets of bounds only differ if there are $\tilde{\lambda}$ such that the composed systems $((1-\tilde{\lambda}) X_{\mathrm{0}}, \tilde{\lambda} X_{\mathrm{1}})$ cannot be reversibly inter-converted with an equilibrium state, meaning that there does not exist a state $\tilde{X} \in \Gamma$ such that $\tilde{X} \sim ((1-\tilde{\lambda}) X_{\mathrm{0}}, \tilde{\lambda} X_{\mathrm{1}})$. In this case the interval $[\tilde{S}_{\mathrm{-}}, \tilde{S}_{\mathrm{+}}]$ 
may be smaller than $[S_{\mathrm{-}}, S_{\mathrm{+}}]$; 
the former may not contain all possible monotonic extensions $S_{\mathrm{ext}}$ of the entropy
function $S$. On the other hand, it may allow for a distinction of non-equilibrium states that lie between the same pair of equilibrium states. Thus, depending on the application one or the other set of bounding quantities may be preferred.
We will see in Section~\ref{sec:results} that the quantities $\tilde{S}_{\mathrm{-}}$ and $\tilde{S}_{\mathrm{+}}$ appear natural in our microscopic setting.
The following lemma confirms that the quantities \eqref{eq:condentr1} and \eqref{eq:condentr2} always lead to the same necessary as well as essentially the same sufficient conditions for state transformations as \eqref{eq:lymin} and \eqref{eq:lymax}.
\begin{lemma} \label{lemma:necsuff}
Let $X,~Y \in \Gamma_{\mathrm{ext}}$. Then, the following two conditions hold:
\begin{equation} \label{eq:c1}
\tilde{S}_{+}(X) < \tilde{S}_{-}(Y) \Rightarrow X \prec Y ,
\end{equation}
\begin{equation} \label{eq:c2}
X \prec Y \Rightarrow \tilde{S}_{-}(X) \leq \tilde{S}_{-}(Y) \text{ and } \tilde{S}_{+}(X) \leq \tilde{S}_{+}(Y).
\end{equation}
\end{lemma}
Note that \eqref{eq:c1} and \eqref{eq:c2} have been shown to hold for $S_{\mathrm{-}}$ and $S_{\mathrm{+}}$ in Ref.~\cite{Lieb2013_entropy_noneq}. The following proof proceeds similarly.
\begin{proof}
To show \eqref{eq:c1}, let $\tilde{S}_{+}(X) < \tilde{S}_{-}(Y)$. Then there exist $\lambda < \lambda'$ such that $X \prec ((1-\lambda)X_0,~\lambda X_1)$ and $((1-\lambda')X_0,~\lambda' X_1) \prec Y$, where $X_0 \prec \prec X_1 \in \Gamma$. 
With E5 we can rewrite 
\begin{equation*}
((1-\lambda)X_0,~\lambda X_1) \sim ((1-\lambda')X_0,~(\lambda'-\lambda)X_0,~\lambda X_1)
\end{equation*}
as well as 
\begin{equation*}
((1-\lambda')X_0,~\lambda' X_1) \sim ((1-\lambda')X_0,~(\lambda'-\lambda)X_1,~\lambda X_1).
\end{equation*}
According to E4, $X_0 \prec X_1$ implies 
$(\lambda'-\lambda)X_0 \prec (\lambda'-\lambda)X_1$.
With E3, the above relations imply
\begin{equation*}
((1-\lambda)X_0, \lambda X_1) \prec ((1-\lambda')X_0, \lambda' X_1).
\end{equation*} 
Thus by transitivity (E2) $X \prec Y$.

For \eqref{eq:c2}, observe that if $((1-\lambda)X_0, \lambda X_1) \prec X$ and $X \prec Y$, then according to E2 also $((1-\lambda)X_0, \lambda X_1)\prec Y$. This implies that $\tilde{S}_{-}(X) \leq \tilde{S}_{-}(Y)$. Similarly, $Y \prec ((1-\lambda')X_0, \lambda' X_1)$ and $X \prec Y$, thus $X \prec ((1-\lambda')X_0, \lambda' X_1)$ and  $\tilde{S}_{+}(X) \leq \tilde{S}_{+}(Y)$.
\end{proof}

Note that in the case of a continuous entropy function $S$ on $\Gamma$, supremum and infimum in the definition of $\tilde{S}_-$ and $\tilde{S}_+$ are attained, which implies that in \eqref{eq:c1} the strict inequality can be replaced by $\leq$. This is in particular satisfied in thermodynamics. For the quantum resource theories considered later supremum and infimum are attained for rational entropy values.

\section{Introduction to quantum resource theories and information-theoretic entropy measures} \label{sec:resources}
In general, a resource theory is a means to investigate which tasks can be achieved or how certain states of systems may change, if the processes affecting a system are of a predefined class, often called the \textit{allowed operations}.

To quantify the utility of different states with respect to a given resource theory, values are assigned to them. This is usually achieved with so called \textit{monotones}, functions that are monotonic under the class of allowed operations. Intuitively, monotones quantify the variety of states that can be generated from each state and serve as a means to compare the latter. A state is called a resource if it cannot be prepared with operations from the allowed class only.

For quantum resource theories, the state space on which the allowed operations act is composed of density operators, i.e., positive semi-definite operators of unit trace on a Hilbert space $\mathcal{H}$. We denote the set of all density operators on $\mathcal{H}$ as $\mathcal{S}(\mathcal{H})$ and consider only finite dimensional Hilbert spaces.
In the following we introduce the two most prominent quantum resource theories, the resource theories of noisy and thermal operations. Note that $\log$ denotes the logarithm with respect to base $2$, while $\ln$ stands for the natural logarithm. 

\subsection{The resource theory of noisy operations} \label{sec:noisy}
The resource theory of noisy operations~\cite{Horodecki2003,Horodecki2013,Gour2013arXiv_resource}
is defined by the following class of allowed operations:
\begin{itemize}
	\item composition of the system with an ancillary system in a maximally mixed state;
	\item reversible transformation on the system and ancilla with any joint unitary;
	\item partial trace over any subsystem.
\end{itemize}
For finite dimensional systems, 
a state $\rho \in \mathcal{S}(\mathcal{H})$ can be transformed into a state $\tilde{\rho} \in \mathcal{S}(\mathcal{H})$ by noisy operations if and only if $\rho$ majorizes $\tilde{\rho}$~\cite{Horodecki2003}.
\begin{definition} \label{definition:majorization} 
Let $\rho$, $\tilde{\rho} \in \mathcal{S}(\mathcal{H})$ with $d = \dim{\mathcal{H}}$ and eigenvalues $p_{1} \geq p_{2} \geq \ldots \geq p_{d}$ and 
$\tilde{p}_{1} \geq \tilde{p}_{2} \geq \ldots \geq \tilde{p}_{d}$. 
The state $\rho$ \textbf{majorizes} the state $\tilde{\rho}$, denoted as $\rho \prec_{\mathrm{M}} \tilde{\rho}$~\footnote{To avoid confusion with Lieb and Yngvason's order relation later, we use a non-standard notation for the majorization relation $\prec_{\mathrm{M}}$,
which in particular differs from the notational convention used in~\cite{BookBhatiaMatrixAnalysis1997}.}, iff 
for all $k \in \left\{ 1,2, \ldots, d \right\}$
\begin{equation*} 
\sum_{i=1}^{k} p_{i} \geq \sum_{i=1}^{k} \tilde{p}_{i},
\end{equation*}
with equality for $k = d$.
\end{definition}
The spectrum of a state $\rho$ with ordered eigenvalues $p_{1} \geq p_{2} \geq \ldots \geq p_{d}$ can be represented as a step function,
\begin{equation} \label{eq:stepfunction}
f_{\mathrm{\rho}}(x) =
\begin{cases} 
p_{i}, &  i-1 \leq x \leq i,\\
0, &  \textrm{otherwise}.
\end{cases} 
\end{equation}
Majorization can be equivalently expressed in terms of $f_{\rho}$:
\begin{equation}
\rho \prec_{\mathrm{M}} \tilde{\rho}   \ \Leftrightarrow \ \int_{0}^{k}f_{\rho}(x) dx \ \geq \ \int_{0}^{k} f_{\tilde{\rho}}(x) dx, 
\quad \forall \ k \in \mathbbm{R}_{\mathrm{\geq 0}}.
\label{eq:majorfunct} \end{equation}
As $f_{\rho}(x)$ is monotonically decreasing in $x$ and due to the normalization $\int_{0}^{\infty}f_{\rho}(x) dx = 1$, 
the condition $k \in \left\{ 1,2, \ldots, d \right\}$ from Definition~\ref{definition:majorization} 
is equivalently replaced with $k \in \mathbbm{R}_{\geq 0}$.
Noisy operations and majorization are furthermore related to unital operations~\cite{Uhlmann1971_Dichtematrizen,Gour2013arXiv_resource,%
Chefles2002,mendl2009unital}.
\begin{definition} A \textbf{unital} quantum operation on $\mathcal{S}(\mathcal{H})$ is a trace preserving completely positive map that preserves the identity operator $\mathbbm{1}_{\dim(\mathrm{\mathcal{H}})}$.
\end{definition}
Note that a positive map $M$ is completely positive if
$M \otimes \mathbbm{1}_{\mathbb{C}^{d \times d}}$ is positive for any $d$, where $\mathbbm{1}_{\mathbb{C}^{d \times d}}$ is the identity on $\mathbb{C}^{d \times d}$.

\begin{prop} \label{prop:unital}
For two density operators $\rho$, $\tilde{\rho} \in \mathcal{S}(\mathcal{H})$ the following are equivalent:
\begin{itemize}
\item There exists a noisy operation that achieves the transformation $\rho \rightarrow \tilde{\rho}$.
\item There exists a unital quantum operation that achieves the transformation
$\rho \rightarrow \tilde{\rho}$.
\item The state $\rho$ majorizes the state $\tilde{\rho}$, denoted $\rho \prec_{\mathrm{M}} \tilde{\rho}$.
\end{itemize}
\end{prop}
This is proven for instance in~\cite{Uhlmann1971_Dichtematrizen,Gour2013arXiv_resource}. Note that majorization is also tightly related to the resource theory of entanglement~\cite{Nielsen1999,Nielsen2001,Nielsen2001b,Horodecki2009}.

The resource theory of noisy operations features numerous monotonic functions~\cite{Gour2013arXiv_resource, Karamata}.
Arguably the most popular monotones are the
R\'{e}nyi entropies~\cite{Renyi1960_MeasOfEntrAndInf}.
\begin{definition} 
The  \textbf{$\alpha$-R\'{e}nyi entropy} of a density operator $\rho \in \mathcal{S}(\mathcal{H})$ is defined as 
\begin{equation*} 
H_{\mathrm{\alpha}}(\rho):=\frac{1}{1-\alpha} \log \operatorname{tr}(\rho^{\alpha}).  \end{equation*}
\end{definition}
The limit $\alpha \rightarrow 1$ yields the von Neumann entropy
$H(\rho)=-\operatorname{tr}( \rho \log \rho)$.
Furthermore, for $\alpha \rightarrow \infty$ and $\alpha=0$ we recover two quantities from the smooth entropy framework, 
the min-entropy, $H_{\mathrm{min}}$, and the max-entropy, $H_{\mathrm{max}}$~\cite{PhdRenner2005_SQKD,PhDTomamichel2012}.
\begin{definition} 
For a density operator $\rho \in \mathcal{S}(\mathcal{H})$ its \textbf{min and max-entropies} are
\begin{align} H_{\mathrm{min}}(\rho) &:= - \log \|\rho\|_{\infty}, \label{eq:hmin-2}\\
H_{\mathrm{max}}(\rho) &:= \log \operatorname{rank}\rho \label{eq:hmax-2}, \end{align}
where $\|\rho\|_{\infty}$ denotes the maximal eigenvalue of $\rho$.%
\footnote{%
  This definition of the max-entropy corresponds to the one in
  Ref.~\cite{PhdRenner2005_SQKD}, and differs from that in
  Refs.~\cite{Koenig2009IEEE_OpMeaning,Tomamichel2009IEEE_AEP,Tomamichel2010IEEE_Duality,PhDTomamichel2012}.
  The former is the R\'enyi entropy of order zero, while the latter is the R\'enyi entropy
  of order $\nicefrac12$.  Note, however, that all R\'enyi entropies with $\alpha<1$
  attain similar values (after smoothing)~\cite{Tomamichel2011TIT_LeftoverHashing}, and
  the term ``max-entropy'' should be understood as designating this class of similar
  measures.
}
\end{definition}

\subsection{The resource theory of thermal operations} \label{sec:thermal}
The resource theory of thermal operations 
describes quantum systems interacting with a thermal environment at a temperature $T$~\cite{Janzing2000_cost,Brandao2013_resource,Horodecki2013_ThermoMaj}.
The allowed operations in this framework are:
\begin{itemize}
	\item composition with an ancillary system (of arbitrary Hamiltonian) in a thermal state relative to the heat bath at temperature $T$;
	\item reversible and energy-conserving transformation of the system and the ancilla;
	\item removal of any subsystem.
\end{itemize} 

Thermal states, also called Gibbs states, are of the form 
\begin{equation} 
\tau = \sum_{i} \frac{e^{-\beta E_i}}{Z}  \left| E_i \right\rangle \left\langle E_i \right|, \label{eq:thermal} 
\end{equation}
where $Z$ is the partition function and the $\left\{\left|E_i\right\rangle\right\}_i$ are the energy eigenstates of a corresponding Hamiltonian $H$; 
the constant $\beta=\frac{1}{k_{\mathrm{B}}T}$ is inversely proportional to the temperature $T$ and $k_{\mathrm{B}}$ denotes the Boltzmann constant. 
These states are preserved under thermal operations, while all athermal states are resource states~\cite{Brandao2015,Horodecki2013_ThermoMaj}.
In the above definition the energy levels may be degenerate. In the extreme case of a system with a trivial Hamiltonian, i.e., a system where all energy levels are degenerate, the resource theory of thermal operations is equivalent to the resource theory of noisy operations.

Closely connected to the resource theory of thermal operations is the order relation of thermo-majorization, which may be defined in terms of a rescaled step function~\cite{Horodecki2013_ThermoMaj,Egloff2012}.
\begin{definition} \label{definition:Gibbs}
Let $\rho \in \mathcal{S}(\mathcal{H})$
be a density matrix that is block diagonal in the energy eigenbasis and let $d=\dim\mathcal{H}$.
Represent its spectrum as in \eqref{eq:stepfunction}.
The \textbf{Gibbs-rescaled step function} of $\rho$ is 
\begin{equation*} 
f^{\mathrm{T}}_{\mathrm{\rho}}(x)= \begin{cases} \frac{p_{\mathrm{i}}}{e^{-\beta E_i}}, &  \sum_ {k=1}^{i-1}  e^{-\beta E_k} \leq x \leq \sum_ {k=1}^{i}  e^{-\beta E_k},\\
 0, &  \textrm{otherwise,} 
\end{cases}
\end{equation*}
where the eigenvalues are reordered such that $\frac{p_{1}}{e^{-\beta E_1}} \geq \frac{p_{2}}{e^{-\beta E_2}} \geq \ldots \geq \frac{p_{d}}{e^{-\beta E_d}}$; $E_i$ denotes the $i$-th energy eigenvalue.\footnote{Note that several energy eigenvalues may be identical.}
\end{definition}
Thermo-majorization may be defined in terms of Gibbs-rescaled step functions, analogous to the formulation of majorization in \eqref{eq:majorfunct}.
\begin{definition} \label{definition:thermom}
Let $\rho$ and $\sigma \in \mathcal{S}(\mathcal{H})$ be two quantum states that are both block diagonal in the energy eigenbasis. 
The order relation of \textbf{thermo-majorization}, $\prec_{\mathrm{T}}$, is defined as
\begin{equation*}
\rho \prec_{\mathrm{T}} \sigma \ \Leftrightarrow \ \int_{0}^{k}f_{\mathrm{\rho}}^{\mathrm{T}}(x) dx \ \geq \ \int_{0}^{k} 
f_{\mathrm{\sigma}}^{\mathrm{T}}(x) dx, \quad \forall \ k \in \mathbbm{R}_{\mathrm{\geq 0}}.
 \end{equation*}
\end{definition}
Note that this definition is equivalent to the one in~\cite{Horodecki2013_ThermoMaj}.
Furthermore, for block diagonal states in the energy eigenbasis,
thermal operations can be characterized in terms of 
$\prec_{\mathrm{T}}$~\cite{Horodecki2013_ThermoMaj}. 
\begin{prop}[Horodecki \& Oppenheim] \label{prop:HO}
Let $\rho$ and $\sigma \in \mathcal{S}(\mathcal{H})$ be two states that are block diagonal in the energy eigenbasis.~\footnote{Note that unitary operations that commute with the Hamiltonian are free operations. They allow for the inter conversion of such states with the corresponding diagonalised ones.} 
Then there exists a thermal operation that achieves the transformation $\rho \rightarrow \sigma$ iff the state $\rho$ thermo-majorizes the state $\sigma$, denoted $\rho \prec_{\mathrm{T}} \sigma$.~\footnote{
Note that the result in~\cite{Horodecki2013_ThermoMaj} is actually more general: it also applies to initial states $\rho$ which are not 
block diagonal in the energy eigenbasis, as long as the final states are.
A thermal operation $\rho \rightarrow \sigma$ is then equivalent to an ordering $\rho_{\mathrm{dec}} \prec_{\mathrm{T}} \sigma$,
where $\rho_{\mathrm{dec}}$ denotes the state $\rho$ decohered in the energy eigenbasis.
The reason for this is that thermal operations commute with a decohering operation in the energy eigenbasis, which is mathematically a projection of the state onto its energy eigenspaces. 
Therefore, the transitions $\rho_{\mathrm{dec}} \rightarrow \sigma$
is equivalent to $\rho \rightarrow\sigma_{\mathrm{dec}}=\sigma$.}
\end{prop}

As shown in~\cite{Brandao2015}, a family of measures that are monotonic under thermal operations\footnote{Note that~\cite{Brandao2015} derived monotonicity for the more general class of thermal operations with catalysis. The results on monotonicity carry over to the subset of thermal operations.} for block diagonal states $\rho$ and for all $\alpha \geq 0$ are 
\begin{equation*} 
F_{\alpha}(\rho) := k_{\mathrm{B}}TD_{\alpha}(\rho||\tau)\ln(2) + F(\tau), 
\end{equation*}
where $\tau$ is the thermal state of the system, $Z_{\mathrm{\tau}}$ is its partition function, and $F(\tau):=-k_{\mathrm{B}}T\ln Z_{\mathrm{\tau}}$ is its free energy. 
For commuting states $\rho$ and $\tau$, the R\'{e}nyi divergences $D_{\alpha}(\rho||\tau)$ for
$\alpha \geq 0$  are
\begin{equation*} 
D_{\alpha}(\rho||\tau) := \frac{1}{\alpha - 1} \log \sum_{i} p_{i}^{\alpha}t_{i}^{1-\alpha},
\end{equation*}
where $p_{i}$ are the eigenvalues of $\rho$ and $t_{i}$ are the eigenvalues of $\tau$.
These measures include the $F_{\mathrm{min}}$ for $\alpha = 0$ and the $F_{\mathrm{max}}$ in the limit $\alpha \rightarrow \infty$: 
\begin{equation} 
F_{\mathrm{min}}(\rho) := k_{\mathrm{B}}TD_{\mathrm{0}}(\rho||\tau) \ln(2) + F(\tau),
\label{eq:fmin} 
\end{equation}
\begin{equation} 
F_{\mathrm{max}}(\rho) := k_{\mathrm{B}}TD_{\mathrm{\infty}}(\rho||\tau) \ln(2) + F(\tau).\label{eq:fmax} 
\end{equation}
Note that $D_{\mathrm{0}}(\rho||\tau)= - \log \operatorname{tr} \Pi_{\mathrm{\rho}} \tau$, where $\Pi_{\mathrm{\rho}}$ is the projector onto the support of $\rho$, and $D_{\mathrm{\infty}}(\rho||\tau)=\log \min\left\{ \lambda: \rho \leq \lambda \tau \right\}$ correspond to the relative entropies with respect to the thermal state of the system~\cite{Datta2009IEEE,Koenig2009IEEE_OpMeaning}. 
For a thermal state all $F_{\alpha}$ coincide.

The two quantities $F_{\mathrm{min}}$ and $F_{\mathrm{max}}$ were originally introduced to describe the extractable work as well as the 
work needed to form a state~\cite{Horodecki2013_ThermoMaj}.
Assuming no errors for these two processes and that the states $\rho \in \mathcal{S}(\mathcal{H})$ 
are block-diagonal in the energy eigenbasis, 
the extractable work under thermal operations is 
\begin{equation*} 
W_{\mathrm{ext}}= F_{\mathrm{min}}(\rho)-F_{\mathrm{min}}(\tau), 
\end{equation*}
whereas the work of formation is
\begin{equation*} 
W_{\mathrm{form}}= F_{\mathrm{max}}(\rho)-F_{\mathrm{max}}(\tau).
\end{equation*}
In the thermodynamic limit the extractable work of a state
$\rho \in \mathcal{S}(\mathcal{H})$ is 
\begin{equation*}
W(\rho)= F(\rho)-F(\tau);
\end{equation*}
the same quantity is used to describe the work of formation~\cite{Horodecki2013_ThermoMaj}.

\section{Application to microscopic systems} \label{sec:results}
The axiomatic framework introduced in Section~\ref{sec:ly} was explicitly devised for macroscopic systems considered from the perspective of phenomenological thermodynamics. We describe these systems from a microscopic viewpoint based on the same axiomatic approach. Section~\ref{sec:model} gives additional details on our microscopic model, while Section~\ref{sec:technical} elaborates on the proof of Proposition~1 from the main text.

\subsection{Description of our microscopic model} \label{sec:model}
In our model, the set of states $\Gamma_{\mathrm{ext}}$ consists of all density operators on a Hilbert space $\mathcal{S}(\mathcal{H})$. All our considerations are restricted to finite dimensional Hilbert spaces $\mathcal{H}$.
The subset of ``equilibrium states'', $\Gamma$, is defined as comprising those states with a
uniform spectrum, i.e., for which all non-zero eigenvalues are equal.

We express the adiabatic processes (as introduced in the main text and detailed also in Section~\ref{sec:ly}) 
by concrete physical operations that consist of the following steps:
\begin{itemize}
\item composition with an extra ancilla system in an ``equilibrium state'';
\item reversible and energy-preserving interaction of the system and the ancilla with a weight system;
\item removal of the ancilla, whose final reduced
state must be the same as its initial state.
\end{itemize}
We denote such processes with the symbol $\stackrel{A}{\rightarrow}$.

As in Lieb and Yngvason's axiomatic framework, there is no restriction on the size of the ancillary system; it may be an environment that is larger than the system of interest and it may even be macroscopic.
The ancilla is removed by tracing it out, where its reduced state has to coincide with its initial one. Thus, the ancilla system, if probed, looks unchanged. No effect is observed by looking at the system's environment (except for the change in the weight system detailed in the following).

The interaction with the weight implements any unitary transformation on the system and the ancilla. 
An explicit model of the weight introduced in Ref.~\cite{Aberg2014PRL_catalytic} represents the idea of an adiabatic process particularly well. 
The weight system $W$ in question has a Hamiltonian
\begin{equation*} 
H_{\mathrm{W}}:= s \sum_{w} w \left| w \right\rangle \left\langle w \right|
\end{equation*} 
corresponding to an energy ladder, where the $\{\left| w \right\rangle\}_{w}$ are orthonormal states 
and the constant $s \in \mathbbm{R}_{\mathrm{\geq 0}}$ defines the energy level spacing of the Hamiltonian.
The weight, assumed to be in a state 
$\sigma = \left| \eta_{\mathrm{L,l_{0}}} \right\rangle \left\langle \eta_{\mathrm{L,l_{0}}} \right|$
with $\left| \eta_{\mathrm{L,l_{\mathrm{0}}}} \right\rangle = \frac{1}{\sqrt{L}}\sum_{w=0}^{L-1} \left| w+l_{\mathrm{0}} \right\rangle$,
is connected to a quantum system $SA$ with Hamiltonian $H_{\mathrm{SA}}$ (in our case $SA$ consists of the system and the ancilla). 
The operations on the system and the weight commute with the weight's translation along the ladder, independently of its absolute energetic state. Only the relative change in energy of the weight, which also corresponds to the energy that is added to or removed from the system and the ancilla, influences the latter.  These operations furthermore conserve the total energy of system, ancilla, and weight combined, in the sense that they commute with $H_{\mathrm{SA}}\otimes \mathbbm{1}_{\mathrm{W}} +\mathbbm{1}_{\mathrm{SA}} \otimes H_{\mathrm{W}}$, where $H_{\mathrm{SA}}$ is the Hamiltonian on system and ancilla and $H_{\mathrm{W}}$ is the Hamiltonian of the weight~\cite{Aberg2014PRL_catalytic,Skrzypczyk2014NComm_individual}.
For large $L$, i.e., for a sufficiently coherent state, these operations on the system $SA$ and the weight allow for the implementation of any unitary operation on $SA$. 
The interaction with the weight thus reduces to the application of arbitrary unitaries to the system and the ancilla, using the weight system catalytically. 

We emphasize that our considerations are, however, not restricted to this particular weight model. For instance in Ref.~\cite{Allahverdyan2004}, an arbitrary potential that is switched on for a time interval is considered instead of an explicit weight system. This model also leads to arbitrary unitary dynamics.

The weight system effectively eliminates the role of energy from the processes on the system and the ancilla or, alternatively, allows us to change their Hamiltonian at will. It is thus not altogether surprising, that the above processes $\stackrel{A}{\rightarrow}$ resemble the noisy operations introduced in Section~\ref{sec:noisy}. However, they allow the composition with any equilibrium ancilla yet require us to return the ancillas to their initial state. The following lemma shows that these processes are also characterized by the majorization relation, and, that they are thus equivalent to noisy operations in the sense that they allow for the same state transformations.
\begin{lemma}
For two states $\rho$, $\sigma \in \mathcal{S}(\mathcal{H}_{\mathrm{S}})$ the following are equivalent:
\begin{itemize}
\item $(A)$: The spectrum of $\rho$ majorizes the spectrum of $\sigma$, i.e., $\rho \prec_{\mathrm{M}} \sigma$.
\item $(B)$: There exists a process $\rho \stackrel{A}{\rightarrow} \sigma$.
\end{itemize}
\end{lemma}

\begin{proof}
(A) $\Rightarrow$ (B):
If $\rho\prec_{\mathrm{M}}\sigma$, then there exists a noisy operation bringing $\rho$ to $\sigma$~\cite{Horodecki2003}, 
i.e., there exists a unitary $U_{\mathrm{SA}}$ acting on an additional ancillary system $A$ such that
\begin{equation*} 
\operatorname{tr}_{\mathrm{A}} \left(U_{\mathrm{SA}} \left( \rho \otimes \frac{\mathbbm{1}_{\mathrm{A}}}{\dim(\mathcal{H}_{\mathrm{A}})} \right) U_{\mathrm{SA}}^{\dagger}\right)=\sigma,
\end{equation*}
where $\frac{\mathbbm{1}_{\mathrm{A}}}{\dim(\mathcal{H}_{\mathrm{A}})} \in \mathcal{S}(\mathcal{H}_{\mathrm{A}})$ is a maximally mixed ancilla and $U_{\mathrm{SA}}^{\dagger}$ denotes the adjoint of $U_{\mathrm{SA}}$. Note that noisy operations allow for the composition with ancilla systems $A$ of arbitrary dimension. In the original proof of this, Horodecki and Oppenheim construct particular operations, for which the unitary $U_{\mathrm{SA}}$ does not change the reduced state on the ancillary system, i.e., 
\begin{equation*} 
\operatorname{tr}_{\mathrm{S}}\left(U_{\mathrm{SA}} \left( \rho \otimes \frac{\mathbbm{1}_{\mathrm{A}}}{\dim(\mathcal{H}_{\mathrm{A}})} \right) U_{\mathrm{SA}}^{\dagger}\right)=
\frac{\mathbbm{1}_{\mathrm{A}}}{\dim(\mathcal{H}_{\mathrm{A}})}.
\end{equation*}
Thus, the ancilla is traced out in the maximally mixed state in which it was added and the process 
is an adiabatic process $\rho \stackrel{A}{\rightarrow} \sigma$ according to our definition.

(B) $\Rightarrow$ (A): As stated in Proposition~\ref{prop:unital}, $\rho \prec_{\mathrm{M}} \sigma$ is equivalent to 
the existence of a unital map from $\rho$ to $\sigma$.
In the following, we show that the processes $\rho \stackrel{A}{\rightarrow} \sigma$ are unital and thus also imply the ordering $\rho \prec_{\mathrm{M}} \sigma$. 
Now let $\chi \in \mathcal{S}(\mathcal{H}_{\mathrm{A}})$ be a state with a flat spectrum, 
i.e., it can be written as $\chi= \sum_{l=1}^{d} \frac{1}{d} \left| l \right\rangle \left\langle l\right|$ with $d \leq \dim(\mathcal{H}_{\mathrm{A}})$. 
Our operations $\stackrel{A}{\rightarrow}$ are the subset of the operations $\operatorname{tr_{\mathrm{A}}} ( U_{\mathrm{SA}} ( (\cdot) \otimes \chi ) U_{\mathrm{SA}}^{\dagger})$ acting on a state $\rho$,
for which the partial trace removes the reduced state $\chi$, 
i.e., $\operatorname{tr}_{\mathrm{S}} \left(U_{\mathrm{SA}} \left( \frac{\mathbbm{1}}{\dim(\mathcal{H}_{\mathrm{S}})} \otimes \chi \right) U_{\mathrm{SA}}^{\dagger} \right)=\chi$; $U_{\mathrm{SA}}$ denotes an arbitrary unitary.

Consider the function $h(\rho)= -\operatorname{tr}(\rho \log \rho)$.\footnote{
Note that even though this function corresponds to the von Neumann entropy, 
we regard it as a purely mathematical function with the mathematical properties that it is subadditive, 
unitarily invariant, additive for product states, 
and it reaches its maximum of $\log \left( \dim(\mathcal{H})\right)$ iff $\rho$ is maximally mixed. 
The reason why we consider this quantity is that we know rather well how it behaves; 
however, any other mathematical function that satisfies these properties could have been used instead.
}
For the maximally mixed state $\rho=\frac{\mathbbm{1}}{\dim(\mathcal{H}_{\mathrm{S}})}$ the following holds,
\begin{align} \begin{split}
h\left(\rho \otimes \chi\right)&=h\left(\frac{\mathbbm{1}}{\dim(\mathcal{H}_{\mathrm{S}})} \otimes \chi \right)\\
&=h\left(U_{\mathrm{SA}} \left( \frac{\mathbbm{1}}{\dim(\mathcal{H}_{\mathrm{S}})} \otimes \chi \right) U_{\mathrm{SA}}^{\dagger}\right)\\
&\leq h\left(\operatorname{tr}_{\mathrm{A}} \left(U_{\mathrm{SA}} \left( \frac{\mathbbm{1}}{\dim(\mathcal{H}_{\mathrm{S}})} \otimes \chi \right) U_{\mathrm{SA}}^{\dagger} \right)\right)\\
&\quad + h\left(\operatorname{tr}_{\mathrm{S}} \left(U_{\mathrm{SA}} \left( \frac{\mathbbm{1}}{\dim(\mathcal{H}_{\mathrm{S}})} \otimes \chi \right) U_{\mathrm{SA}}^{\dagger} \right)\right).\\
\end{split} \label{eq:h} \end{align}
The inequality follows by subadditivity of $h$.
Furthermore, we know that
\begin{equation*} 
h\left(\frac{\mathbbm{1}}{\dim(\mathcal{H}_{\mathrm{S}})} \otimes \chi\right)=\log( \dim(\mathcal{H}_{\mathrm{S}}) \cdot d).
\end{equation*}
Since the ancillary system has to be in state $\chi$ at the end of the process,
\begin{equation*} h\left(\operatorname{tr}_{\mathrm{S}} \left(U_{\mathrm{SA}} \left( \frac{\mathbbm{1}}{\dim(\mathcal{H}_{\mathrm{S}})} \otimes \chi \right) U_{\mathrm{SA}}^{\dagger} \right)\right)=\log d.
\end{equation*}
From \eqref{eq:h} we conclude that
\begin{equation*}
h\left(\operatorname{tr}_{\mathrm{A}} \left(U_{\mathrm{SA}} \left( \frac{\mathbbm{1}}{\dim(\mathcal{H}_{\mathrm{S}})} \otimes \chi \right) U_{\mathrm{SA}}^{\dagger} \right)\right) \geq \log \dim(\mathcal{H}_{\mathrm{S}}).
\end{equation*}
Note that $\operatorname{tr}_{\mathrm{A}} \left(U_{\mathrm{SA}} \left( \frac{\mathbbm{1}}{\dim(\mathcal{H}_{\mathrm{S}})} \otimes \chi \right) U_{\mathrm{SA}}^{\dagger} \right) \in \mathcal{S}(\mathcal{H})$, as the partial trace itself is a trace preserving and completely positive map.
Therefore, the above inequality can only be satisfied if $\operatorname{tr}_{\mathrm{A}} \left(U_{\mathrm{SA}} \left( \frac{\mathbbm{1}}{\dim(\mathcal{H}_{\mathrm{S}})} \otimes \chi \right) U_{\mathrm{SA}}^{\dagger} \right)=\frac{\mathbbm{1}}{\dim(\mathcal{H}_{\mathrm{S}})}$, as otherwise its entropy is smaller. The processes $\stackrel{A}{\rightarrow}$ are thus unital.
\end{proof}

The majorization relation imposes a partial ordering on the set of all quantum states, as required by Lieb and Yngvason's framework. Furthermore, not all quantum states can be compared by means of the majorization relation, thus the latter emerges paired with a notion of comparability that expresses whether two states can be ordered by means of $\prec_{\mathrm{M}}$.

Recall that for two states $\rho$, $\tilde{\rho} \in \mathcal{S}(\mathcal{H})$ the majorization condition $\rho \prec_{\mathrm{M}} \tilde{\rho}$
can be equivalently expressed in terms of the spectral step functions $f_{\rho}$ and $f_{\tilde{\rho}}$ according to \eqref{eq:majorfunct}.
Even though the functions $f_{\rho}$ are not in one-to-one correspondence with the states $\rho \in \mathcal{S}(\mathcal{H})$ (they rather represent all states with a certain spectrum) this description is sufficient for this purpose.\footnote{As the processes $\stackrel{A}{\rightarrow}$ include the application of an arbitrary unitary, states with the same spectra can be inter converted by means of an allowed operation.}

For an equilibrium state $\rho$, which, by definition, has a flat spectrum, the step function $f_{\mathrm{\rho}}$ is conveniently written in terms of its rank as
\begin{equation*} 
f_{\mathrm{\rho}}(x) =
	\begin{cases} 
	\frac{1}{\operatorname{r}(\rho)}, &  0 \leq x \leq \operatorname{r}(\rho),\\
	 0, & \textrm{otherwise},
\end{cases}
\end{equation*}
with  $\operatorname{r}(\rho):= \operatorname{rank}\rho$. This notation for the $\operatorname{rank}$ is used throughout this section.
For two equilibrium states $\rho$ and $\tilde{\rho}$ the rank alone determines 
which one majorizes the other, as in that case the relation $\rho \prec_{\mathrm{M}} \tilde{\rho}$ is equivalent to
\begin{equation*}
\int_{0}^{k} \frac{1}{\operatorname{r}(\rho)} dx \ \geq 
\ \int_{0}^{k} \frac{1}{\operatorname{r}(\tilde{\rho})} dx, \quad \forall \ k \in \mathbbm{R}_{\mathrm{\geq 0}}, 
\end{equation*}
and thus to $\operatorname{r}(\rho) \leq \operatorname{r}(\tilde{\rho})$.

We define the composition of two states $\rho \in \mathcal{S}(\mathcal{H})$ and $\tilde{\rho} \in \mathcal{S}(\mathcal{\tilde{H}})$ as their tensor product $\rho \otimes \tilde{\rho} \in \mathcal{S}(\mathcal{H} \otimes \mathcal{\tilde{H}})$.
This means that two systems are considered uncorrelated before any interaction takes place. In an interaction correlations between subsystems may be established. Similarly, the ancillary system is assumed to be initially uncorrelated with the system, which can be ensured by choosing ancillas that have not interacted with the system before or that have been decoupled from the system through other interactions. The same view is taken in~\cite{Horodecki2013_ThermoMaj,Aberg2013_worklike}.
The scaling of an equilibrium state $\rho$ is assumed to coincide with its composition for scaling factors $\lambda \in \mathbbm{N}$.
The step function of a scaled state $\lambda \rho \in \lambda \mathcal{S}(\mathcal{H}$), 
which is defined as $\rho^{\otimes \lambda} \in \mathcal{S}(\mathcal{H}^{\otimes \lambda})$, is
\begin{align*}
  \begin{split}
    f_{ \lambda \rho}(x) &= 
    \begin{cases}
      \left(\frac{1}{ \operatorname{r}(\rho)}\right)^{\lambda}, &  0 \leq x \leq \operatorname{r}(\rho)^{\lambda},\\
      0, &  \textrm{otherwise},
    \end{cases} \\
    &= f_{\rho}(x^{\frac{1}{\lambda}})^{\lambda}.
  \end{split}
\end{align*}
Note that correlations within the system do not affect this scaling operation: From the perspective of an outside observer with the ability to apply adiabatic processes to the system, its internal properties (such as the entanglement of individual particles, for instance) are not accessible. Hence, they may not be incorporated into the observer's description of the system and our considerations escape the potential issues
with entanglement in microscopic systems that were raised in~\cite{Lieb2014}.

To obtain a continuous scaling operation, this way of scaling the step function is applied for any $\lambda \in \mathbbm{R}_{> 0}$. 
For most values of $\lambda$ the scaled copies $\lambda \rho$ do not represent physical states and no actual space 
$\lambda \mathcal{S}(\mathcal{H})$ exists.
However, any normalized, but possibly unphysical, function $f(x)$, can be turned into
a physical state by actually considering the function $\frac1n f(x/n)$ for a large enough
$n$: this function now effectively describes the state of the system along with an ancilla of dimension $n$ in a fully mixed state.
Indeed, we can always combine states with a fully mixed state of a given rank, and
 the following rules apply:
\begin{subequations}
  \begin{align*}
    \lambda\left(\rho,\frac{\mathbbm{1}_n}{n}\right) &\sim
    \left(\lambda\rho,\frac{\mathbbm{1}_{n^\lambda}}{n^\lambda}\right) \ ;
    \\
    \left(\left(\rho,\frac{\mathbbm{1}_n}{n}\right),
      \left(\tilde{\rho},\frac{\mathbbm{1}_m}{m}\right)\right)
    &\sim \left(\left(\rho, \tilde{\rho}\right) , \frac{\mathbbm{1}_{n\cdot m}}{n\cdot m}\right)\ .
  \end{align*}
\end{subequations}
These rules are easily verified by representing the states in terms of their step functions. 
They give physical significance to any relative statements involving states that
would be required to be scaled in an unphysical way: for example,
$(\lambda\rho,\mu\tilde{\rho})\prec\sigma$ if and only if
$\left(\lambda\left(\rho,\frac{\mathbbm{1}_n}{n}\right),
  \mu\left(\tilde{\rho},\frac{\mathbbm{1}_m}{m}\right)\right) \prec
\left(\sigma,\left(\frac{\mathbbm{1}_{n^\lambda m^\mu}}{n^\lambda m^\mu}\right)\right)$.
Thus, if $\lambda\rho$ does not actually correspond to a physical state, then the second
expression should in fact be considered; indeed for large enough $n$ the state
$\lambda\left(\rho,\frac{\mathbbm{1}_n}{n}\right)$ is physical to a good approximation.

The entropy function~\eqref{eq:entropyoriginal1} can thus be equivalently rewritten as
\begin{widetext}
\begin{equation} \label{eq:entropyscaled} 
S\left(\rho\right) 
= \sup\left\{ \lambda : \left( \left(1-\lambda\right)
\left(\rho_0,~\frac{\mathbbm{1}_n}{n}\right),~\lambda
\left(\rho_1,~\frac{\mathbbm{1}_n}{n}\right)\right) \prec \left(\rho,~\frac{\mathbbm{1}_{n}}{n}\right) \right\}\ ,
\end{equation}
\end{widetext}
where the last expression (to a good approximation) only involves physical states for a large enough $n$. 

\subsection{Proof of Proposition~1 from the main text} \label{sec:technical}
In order to derive entropy functions for our microscopic model, we first have to show that it obeys Lieb and Yngvason's axioms.
\begin{lemma} \label{lemma:major}
Consider the majorization relation $\prec_{\mathrm{M}}$
and the composition and scaling operations defined above.
Then, for equilibrium states the six axioms E1 to E6 as well as the Comparison Hypothesis hold. For non-equilibrium states axioms N1 and N2 hold.
\end{lemma}
\begin{proof}
As axiom N2 requires that axioms E1 to E3 as well as E6 hold for non-equilibrium states, 
we directly verify them for arbitrary states.
\newline
\textit{Reflexivity (E1):} $\prec_{\mathrm{M}}$ is clearly reflexive: For $\rho \in \mathcal{S}(\mathcal{H})$ and for all $k \in \mathbbm{R}_{\geq 0}$,
\begin{equation*} 
\int^{k}_{0} f_{\rho}(x) dx = \int^{k}_{0} f_{\rho}(x) dx, 
\end{equation*}
and thus $\rho \sim_{\mathrm{M}} \rho$.
\newline
\textit{Transitivity (E2):} Let $\rho$, $\sigma$, and $\chi \in \mathcal{S}(\mathcal{H})$. Then, $\rho \prec_{\mathrm{M}} \sigma$ and $\sigma \prec_{\mathrm{M}} \chi$ is equivalent to
\begin{equation*} 
\int^{k}_{0} f_{\rho}(x) dx \geq \int^{k}_{0} f_{\sigma}(x) dx \geq \int^{k}_{0} f_{\chi}(x) dx, \quad \forall \ k \in \mathbbm{R}_{\geq 0}.
\end{equation*}
Thus, $\rho \prec_{\mathrm{M}} \chi$.
\newline
\textit{Consistent composition (E3):} Let $\rho \prec_{\mathrm{M}} \sigma \in \mathcal{S}(\mathcal{H})$ with $d=\dim(\mathcal{H})$ 
and with eigenvalues
$p_{1} \geq p_{2} \geq \ldots \geq p_{d}$ and $q_{1} \geq q_{2} \geq \ldots \geq q_{d}$
and let $\rho' \prec_{\mathrm{M}} \sigma' \in \mathcal{S}(\mathcal{H'})$ with $d'=\dim(\mathcal{H'})$ and with eigenvalues
$p'_{1} \geq p'_{2} \geq \ldots \geq p'_{d'}$ and $q'_{1} \geq q'_{2} \geq \ldots \geq q'_{d'}$. 
Note that the composed state $(\rho,~\rho')$ with eigenvalues $p''_{l}=p_{i} p'_{j}$ ordered according to
$p''_{1} \geq p''_{2} \geq \ldots \geq p''_{d d'}$ has an associated step function
\begin{equation*}
f_{\mathrm{(\rho,~\rho')}}(x)= 
\begin{cases} 
p''_{l}, &  l-1 \leq x \leq l,\\
0, &  \textrm{otherwise}, 
\end{cases},
\end{equation*}
for $l=1,~2,~\ldots~,~d d'$.
Then for any $0 \leq k \leq d d'$ we can choose $m_{1} \geq m_{2} \geq \ldots \geq m_{d'} \geq 0$ with $m_{1}+m_{2}+ \cdots + m_{d'} =k$ such that
\begin{equation*}
\sum_{l=0}^{k} p''_{l} 
= p_1 \sum_{j=0}^{m_{1}} p'_{j} + p_2 \sum_{j=0}^{m_{2}} p'_{j} + \cdots + p_d \sum_{j=0}^{m_{d'}} p'_{j}. 
\end{equation*}
This is possible since for any $i$ and $j$ the inequality $p_i p'_j \geq p_i p'_{j+k }$ holds for all $k \geq 0$.
Now as $\rho' \prec_{\mathrm{M}} \sigma'$ the inequality
\begin{equation*}
\sum_{j=0}^{m_{i}} p'_{j} \leq \sum_{j=0}^{m_{i}} q'_{j},
\end{equation*}
holds for any $0 \leq i \leq d'$. Therefore,
\begin{multline}
p_1 \sum_{j=0}^{m_{1}} p'_{j} + p_2 \sum_{j=0}^{m_{2}} p'_{j} + \cdots + p_d \sum_{j=0}^{m_{d'}} p'_{j} \\
\leq
 p_1 \sum_{j=0}^{m_{1}} q'_{j} + p_2 \sum_{j=0}^{m_{2}} q'_{j} + \cdots + p_d \sum_{j=0}^{m_{d'}} q'_{j}. \label{eq:inequ}
\end{multline}
Now the terms in~\eqref{eq:inequ} can be regrouped with
$n_{1} \geq n_{2} \geq \ldots \geq n_{d} \geq 0$ obeying $n_{1}+n_{2}+ \cdots + n_{d} =k$ as 
\begin{multline*}
q'_1 \sum_{j=0}^{n_{1}} p_{j} + q'_2 \sum_{j=0}^{n_{2}} p_{j} + \cdots + q'_{d'} \sum_{j=0}^{n_{d'}} p_{j} \\
\leq q'_1 \sum_{j=0}^{n_{1}} q_{j} + q'_2 \sum_{j=0}^{n_{2}} q_{j} + \cdots + q'_{d'} \sum_{j=0}^{n_{d'}} q_{j}.
\end{multline*}
The inequality holds as $\rho \prec_{\mathrm{M}} \sigma$ implies that for any $0 \leq i \leq d$,
\begin{equation*}
\sum_{j=0}^{m_{i}} p_{j} \leq \sum_{j=0}^{m_{i}} q_{j}.
\end{equation*}
Furthermore,
\begin{equation*}
q'_1 \sum_{j=0}^{n_{1}} q_{j} + q'_2 \sum_{j=0}^{n_{2}} q_{j} + \cdots + q'_{d'} \sum_{j=0}^{n_{d'}} q_{j} \leq \sum_{l=0}^{k} q''_l.
\end{equation*}
as the terms in the last sum might be ordered differently. This concludes the proof that $(\rho,~\rho') \prec_{\mathrm{M}} (\sigma,~\sigma')$.
\newline
\textit{Stability (E6):} Note that $\chi_{\mathrm{0}}$ and $\chi_{\mathrm{1}}$ are equilibrium states; 
only for those the scaling operation is defined.
Let $\rho \in \mathcal{S}(\mathcal{H})$ have eigenvalues $p_{1} \geq p_{2} \geq \ldots \geq p_{d}$ with $d=\dim(\mathcal{H})$.
Assume that $ (\rho,~\varepsilon \chi_{\mathrm{0}} ) \prec_{\mathrm{M}} (\sigma,~\varepsilon \chi_{\mathrm{1}})$
for a sequence of $\varepsilon$'s tending to zero. Let this sequence be denoted by $(\varepsilon_{i})_{i}$.
Then for all $\varepsilon_{i}$ we have
\begin{equation*} 
\int_{0}^{k} f_{(\rho, \varepsilon_{i} \chi_{\mathrm{0}})}(x) dx \ \geq \ \int_{0}^{k} f_{\mathrm{(\sigma, \varepsilon_{i} \chi_{\mathrm{1}})}}(x) dx, \quad \forall \ k \in \mathbbm{R}_{\mathrm{\geq 0}}. 
\end{equation*}
Taking the limit $i \rightarrow \infty$ leads to
\begin{equation*}
  \lim_{i\rightarrow \infty} \int_{0}^{k} f_{\mathrm{(\rho,\varepsilon_{i}
      \chi_{0})}}(x) dx \ \geq \ \lim_{i\rightarrow \infty}  \int_{0}^{k}
  f_{\mathrm{(\sigma,\varepsilon_{i}\chi_{1})}}(x) dx,
\end{equation*}
$\forall \ k \in \mathbbm{R}_{\mathrm{\geq 0}}$ and thus
\begin{equation*}
\int_{0}^{k} \lim_{i\rightarrow \infty} f_{\mathrm{(\rho,\varepsilon_{i}\chi_{0})}}(x) dx
\geq  \int_{0}^{k} \lim_{i\rightarrow \infty}
f_{\mathrm{(\sigma,\varepsilon_{i}\chi_{1})}}(x) dx,
\end{equation*}
$\forall \ k \in \mathbbm{R}_{\mathrm{\geq 0}}$, which follows by dominated convergence.
Note that 
\begin{equation*}
f_{\mathrm{(\rho,\varepsilon_{i} \chi_{\mathrm{0}})}}(x)= 
\begin{cases} 
p_{l} \frac{1}{\operatorname{r}(\chi_{0})^{\varepsilon_{i}}},  &  (l-1)\operatorname{r}(\chi_{0})^{\varepsilon_{i}}  \leq x \leq l \operatorname{r}(\chi_{0})^{\varepsilon_{i}},\\
0, &  \textrm{otherwise}, 
\end{cases}
\end{equation*}
and similarly for $f_{\mathrm{(\sigma,\varepsilon_{i} \chi_{\mathrm{1}})}}$.
Thus, taking the limit leads to 
\begin{equation*} 
\int_{0}^{k}f_{\mathrm{\rho}}(x) dx \ \geq \ \int_{0}^{k} f_{\mathrm{\sigma}}(x) dx, \quad \forall \ k \in \mathbbm{R}_{\mathrm{\geq 0}},
\end{equation*} 
which is equivalent to $\rho \prec_{\mathrm{M}} \sigma$. 

Since for equilibrium states the rank suffices to assert majorization, E4 and E5 as well as the Comparison Hypothesis can be efficiently checked:
\newline
\textit{Scaling invariance (E4):} We see that
\begin{align*}\begin{split}
\rho \prec_{\mathrm{M}} \sigma \ &\Leftrightarrow \  \operatorname{r}(\sigma) \ \geq \operatorname{r}(\rho) \\
&\Leftrightarrow \ \operatorname{r}(\sigma)^{\lambda} \ \geq \ \operatorname{r}(\rho)^{\lambda}, \quad \forall \ \lambda \ > \ 0 \\
&\Leftrightarrow \ \lambda \rho \prec_{\mathrm{M}} \lambda \sigma, \quad \forall \ \lambda \ > \ 0.
\end{split} 
\end{align*}
\newline
\textit{Splitting and recombination (E5):} Let $0 < \lambda < 1$. Then,
\begin{align*}
\begin{split}
f_{\mathrm{(\lambda \rho, (1-\lambda) \rho)}}(x)&=\begin{cases} 
\frac{1}{\operatorname{r}(\rho)^{\lambda}} \frac{1}{\operatorname{r}(\rho)^{1-\lambda}}, & 0 \leq x \leq \operatorname{r}(\rho)^{\lambda}\operatorname{r}(\rho)^{1-\lambda},\\
	 0, & \textrm{otherwise},
	 \end{cases}\\
&=\begin{cases} 
\frac{1}{\operatorname{r}(\rho)}, & 0 \leq x \leq \operatorname{r}(\rho),\\
	 0, & \textrm{otherwise},
	 \end{cases}\\
&=f_{\mathrm{\rho}}(x).
\end{split}
\end{align*}
\newline
\textit{Comparison Hypothesis:} 
Let $0 \leq \lambda \leq 1$ and $0 \leq \mu \leq 1$ and let $\rho$, $\tilde{\rho}$, $\sigma$ $\tilde{\sigma} \in \mathcal{S}(\mathcal{H})$ 
be arbitrary equilibrium states.
Then the state $(\lambda \rho,  (1-\lambda) \sigma)$ 
with step function
\begin{equation*} 
f_{\mathrm{ (\lambda \rho,(1-\lambda)\sigma)}}(x) = \begin{cases}
 \frac{1}{\operatorname{r}(\rho)^{\lambda}} 
 \frac{1}{\operatorname{r}(\sigma)^{1-\lambda}} , &  0 \leq x \leq \operatorname{r}(\rho)^{\lambda} \operatorname{r}(\sigma)^{1-\lambda},\\
	 0, &  \textrm{otherwise},
\end{cases}
\end{equation*}
can be related to any $(\mu \tilde{\rho},  (1-\mu) \tilde{\sigma})$ 
in the sense of~\eqref{eq:majorfunct}. 
\newline
\textit{Axiom (N1):} Choose $\rho_{\mathrm{0}}$ to be a pure state and $\rho_{\mathrm{1}} = \frac{\mathbbm{1}}{\dim(\mathcal{H})}$ maximally mixed. 
For any state $\rho \in \mathcal{S}(\mathcal{H})$ we know that $\rho_{\mathrm{0}} \prec_{\mathrm{M}} \rho$ and $\rho \prec_{\mathrm{M}} \rho_{\mathrm{1}}$.
\end{proof}

\begin{cor} \label{cor:result}
For a system on which processes $\stackrel{A}{\rightarrow}$ can be performed, 
there exists a unique entropy function $S$ for equilibrium states, 
as well as bounds $S_{\mathrm{-}}$ and $S_{\mathrm{+}}$ on the entropy of non-equilibrium states,
all up to an affine change of scale.
\end{cor}
\begin{proof}
Lemma~\ref{lemma:major} allows the application of Lieb and Yngvason's Theorem~\ref{thm:entropythm} and Proposition~\ref{prop:nonequ}
to quantum states compared by means of $\prec_{\mathrm{M}}$. This implies the corollary.
\end{proof}

We are now in a position to prove Proposition~1 from the main text.

\begin{proof}[Proof of Proposition~1 from the main text]
For any state $\rho \in \mathcal{S}(\mathcal{H})$ with eigenvalues $p_{1} \geq p_{2} \geq \ldots \geq p_{d}$, where $d=\dim \mathcal{H}$, and for any
equilibrium states $\rho_{0} \prec\prec_{\mathrm{M}} \rho_{1}$, let $\lambda$ be such that $((1-\lambda)\rho_{\mathrm{0}},\lambda \rho_{\mathrm{1}}) \prec_{\mathrm{M}} \rho$. Then,
\begin{equation} 
\int_{0}^{k} f_{((1-\lambda) \rho_{\mathrm{0}}, \lambda\rho_{\mathrm{1}})}(x)dx 
\geq \int_{0}^{k} f_{\mathrm{\rho}}(x)dx, \quad \forall \ k \in \mathbbm{R}_{\mathrm{\geq 0}}.
\label{eq:58}
\end{equation}
Let $\tilde{k}=\operatorname{r}(\rho_{\mathrm{0}})^{\mathrm{1-\lambda}}\operatorname{r}(\rho_{\mathrm{1}})^{\mathrm{\lambda}}$.
As $\rho_{\mathrm{0}}$ and $\rho_{\mathrm{1}}$ are equilibrium states, we know that for 
$0 \leq k \leq \tilde{k}$,
\begin{equation*} 
\int_{0}^{k} f_{((1-\lambda)\rho_{\mathrm{0}}, \lambda\rho_{\mathrm{1}})}(x)dx \\= \int_{0}^{k} \left(\frac{1}{\operatorname{r}(\rho_{\mathrm{0}})} \right)^{\mathrm{1-\lambda}} 
\left(\frac{1}{\operatorname{r}(\rho_{\mathrm{1}})}  \right)^{\mathrm{\lambda}} dx.
\end{equation*}
Thus, by considering $k \leq 1$, Equation \eqref{eq:58} directly implies 
\begin{equation}
\left(\frac{1}{\operatorname{r}(\rho_{\mathrm{0}})}  \right)^{\mathrm{1-\lambda}} \left(\frac{1}{\operatorname{r}(\rho_{\mathrm{1}})}  \right)^{\mathrm{\lambda}} \geq p_{1}
,\label{eq:p1cond} 
\end{equation}
which can be rewritten as
\begin{equation}
C_1 \cdot \log \frac{1}{p_{1}}+ C_0 \geq \lambda,\label{eq:abform}
\end{equation}
with $C_1=\frac{1}{\log \frac{\operatorname{r}(\rho_{\mathrm{1}})}{\operatorname{r}(\rho_{\mathrm{0}})}}$ and 
$C_0=-\frac{1}{\log \frac{\operatorname{r}(\rho_{\mathrm{1}})}{\operatorname{r}(\rho_{\mathrm{0}})}} \log \operatorname{r}(\rho_{\mathrm{0}})$. 
On the other hand, Equation \eqref{eq:p1cond} implies
\begin{align*} 
\left(\frac{1}{\operatorname{r}(\rho_{\mathrm{0}})}
  \right)^{\mathrm{1-\lambda}} \left(\frac{1}{\operatorname{r}(\rho_{\mathrm{1}})}
  \right)^{\mathrm{\lambda}} \min \left\{k, \tilde{k} \right\} 
  &\geq p_{1} \cdot \min \left\{k, \frac{1}{p_{1}}\right\} \\
  &\geq \int_{0}^{k} f_{\mathrm{\rho}}(x)dx, 
\end{align*}
for all $k \in \mathbbm{R}_{\mathrm{\geq 0}}$, i.e., implies \eqref{eq:58}; 
the second inequality follows as the step function $f_{\rho}(x)$ is monotonously decreasing and normalised.
Thus, taking the supremum over $\lambda$ in \eqref{eq:abform} concludes the proof for 
$\tilde{S}_{\mathrm{-}}$, as $H_{\mathrm{min}}(\rho)= - \log \|\rho \|_{\infty}$, 
where $\|\rho\|_{\infty}$ denotes the maximal eigenvalue of the state $\rho$, i.e., equals $p_{1}$. 

For $\tilde{S}_{\mathrm{+}}$ we proceed similarly. 
Let $\rho \in \mathcal{S}(\mathcal{H})$ 
and let $\rho_{0} \prec\prec_{\mathrm{M}} \rho_{1}$ be two
equilibrium states, as above. 
Now let $\lambda$ be such that $\rho \prec_{\mathrm{M}} ((1-\lambda)\rho_{\mathrm{0}},\lambda \rho_{\mathrm{1}})$ and thus
\begin{equation} 
\int_{0}^{k} f_{((1-\lambda)\rho_{\mathrm{0}}, \lambda\rho_{\mathrm{1}})}(x)dx
\leq  \int_{0}^{k} f_{\mathrm{\rho}}(x)dx,  \quad \forall \ k \in \mathbbm{R}_{\mathrm{\geq 0}}. 
\label{eq:74}
\end{equation}
First, we show that 
\begin{equation}
\operatorname{r}(\rho_{\mathrm{0}})^{\mathrm{1-\lambda}} \operatorname{r}(\rho_{\mathrm{1}})^{\lambda} \geq \operatorname{r}(\rho) \label{eq:suff}
\end{equation}
by contradiction: Assume for now that $\operatorname{r}(\rho_{\mathrm{0}})^{\mathrm{1-\lambda}} \operatorname{r}(\rho_{\mathrm{1}})^{\lambda} < \operatorname{r}(\rho)$.
For $\tilde{k} = \operatorname{r}(\rho_{\mathrm{0}})^{\mathrm{1-\lambda}} \operatorname{r}(\rho_{\mathrm{1}})^{\lambda}$,
\begin{equation*} 
\int_{0}^{\tilde{k}} f_{\mathrm{\rho}}(x)dx < 1.
\end{equation*} 
This contradicts~\eqref{eq:74}, as 
\begin{equation*}
\int_{0}^{\tilde{k}} f_{((1-\lambda)\rho_{\mathrm{0}}, \lambda\rho_{\mathrm{1}})}(x)dx = 1.
\end{equation*}
Thus we have $\tilde{k} \geq \operatorname{r}(\rho)$, which can be rewritten as 
\begin{equation} \lambda \geq C_1 \cdot \log \operatorname{r}(\rho) + C_0 \label{eq:abcond2}\end{equation}
with $C_1$ and $C_0$ defined as above.
Moreover, \eqref{eq:suff} implies 
\begin{align*}
\left(\frac{1}{\operatorname{r}(\rho_{\mathrm{0}})} \right)^{\mathrm{1-\lambda}} \left(\frac{1}{\operatorname{r}(\rho_{\mathrm{1}})} \right)^{\mathrm{\lambda}} 
\min \left\{k, \tilde{k} \right\}  
&\leq \frac{1}{\operatorname{r}(\rho)} \min \left\{k, \operatorname{r}(\rho) \right\} \\
&\leq \int_{0}^{k} f_{\mathrm{\rho}}(x)dx
\end{align*}
for all $k \in \mathbbm{R}_{\geq 0}$, i.e., implies \eqref{eq:74}; 
the second inequality holds as $f_{\rho}(x)$ is monotonously decreasing and normalised.
As $H_{\mathrm{max}}(\rho)= \log \operatorname{r}(\rho)$, taking the infimum over $\lambda$ in \eqref{eq:abcond2} concludes the proof for $\tilde{S}_{+}$.

In the case of an equilibrium state $\rho$,
for $\lambda = S(\rho)$ the relation $((1-\lambda)\rho_{\mathrm{0}}, \lambda \rho_{\mathrm{1}}) \sim_{\mathrm{M}} \rho$ holds. Note that for these states $\log \operatorname{r}(\rho)=\log \frac{1}{p_1}$. Thus, according to the above considerations,
\begin{equation*}
\lambda=C_1 \cdot \log \operatorname{r}(\rho) + C_0,
\end{equation*}
with parameters $C_0$ and $C_1$ defined like before. This also coincides with the von Neumann entropy $H$ (up to an affine change of scale).
\end{proof}
In particular, $H_{\mathrm{min}}$ and $H_{\mathrm{max}}$ are exactly recovered for the choice
$\operatorname{r}(\rho_{\mathrm{0}})=1$ and $\operatorname{r}(\rho_{\mathrm{1}})=2$. States
obeying $\rho_{1} \prec_{\mathrm{M}} \rho$ are associated with a $\lambda > 1$ and thus an entropy larger than $1$, interpreted like in Section~\ref{sec:ly}.

According to Lemma~\ref{lemma:necsuff}, the information-theoretic entropy measures $H_{\mathrm{min}}$ and $H_{\mathrm{max}}$ obtain an additional meaning as characterising the state transformations by means of adiabatic processes in terms of necessary and sufficient conditions.
Note that these entropy measures and thus the $\tilde{S}_{\mathrm{-}}$ and $\tilde{S}_{\mathrm{+}}$ in our particular model have the appealing property that they are additive under the composition of two states $\rho$, $\sigma \in \mathcal{S}(\mathcal{H})$:
\begin{align*} 
\tilde{S}_{\mathrm{-}}(\rho \otimes \sigma)&=\tilde{S}_{\mathrm{-}}(\rho)+\tilde{S}_{\mathrm{-}}(\sigma), \\
\tilde{S}_{\mathrm{+}}(\rho \otimes \sigma)&=\tilde{S}_{\mathrm{+}}(\rho)+\tilde{S}_{\mathrm{+}}(\sigma). \end{align*}
Thus, if $\sigma$ is an equilibrium state,
\begin{align*} 
\tilde{S}_{\mathrm{-}}(\rho \otimes \sigma)&=\tilde{S}_{\mathrm{-}}(\rho)+S(\sigma), \\
\tilde{S}_{\mathrm{+}}(\rho \otimes \sigma)&=\tilde{S}_{\mathrm{+}}(\rho)+S(\sigma).
\end{align*}

As mentioned in Section~\ref{sec:ly}, these bounds $\tilde{S}_{\mathrm{-}}$ and $\tilde{S}_{\mathrm{+}}$ may not equal $S_{\mathrm{-}}$ and $S_{\mathrm{+}}$. With the continuous scaling operation and our understanding of the entropic bounds in the same way as previously explained for the entropy function \eqref{eq:entropyscaled}, the two sets of entropic quantities coincide.

This is further illustrated with the following example of a qubit. Let $\rho_{\mathrm{0}}= \left|0\right\rangle\left\langle  0 \right|$ and $\rho_{\mathrm{1}}=\frac{1}{2} \left|0\right\rangle\left\langle 0 \right|+\frac{1}{2} \left|1 \right\rangle\left\langle 1 \right|$.
Consider a qubit in the state  $\rho=\frac{3}{4} \left|0\right\rangle\left\langle 0 \right|+\frac{1}{4} \left|1 \right\rangle\left\langle 1 \right|$.
Now the closest qubit equilibrium state $\rho'$ preceding it in the order is $\rho'=\rho_{\mathrm{0}}$, a state with entropy $S(\rho')=0$. With the continuous scaling operation, as detailed above, we actually consider
\begin{align*}
\sup \left\{ S(\rho'): \left(\rho', \frac{\mathbbm{1}}{3} \right) \prec \left(\rho, \frac{\mathbbm{1}}{3} \right) \right\}
= \log \frac{4}{3},
\end{align*}
where not $\rho'$ but rather $\left(\rho', \frac{\mathbbm{1}}{3} \right)$ is required to be a physical state. This equals 
\begin{equation*}
\tilde{S}_{\mathrm{-}}(\rho)=\sup \left\{ \lambda: ((1-\lambda) \rho_{\mathrm{0}}, \lambda \rho_{\mathrm{1}}) \prec \rho \right\}=\log \frac{4}{3}.
\end{equation*} 
Note that we may thus distinguish the state $\rho$ from another state $\sigma=\frac{2}{3} \left|0\right\rangle\left\langle 0 \right|+\frac{1}{3} \left|1 \right\rangle\left\langle 1 \right|$ by means of its entropic bounds.

\section{Interaction with a heat bath} \label{sec:results2}
Lieb and Yngvason's framework is based on rather abstract axioms, which admits its application to other physical contexts.
In typical laboratory experiments, the systems of interest interact with a thermal environment. 
Considering systems connected to an external heat reservoir is thus a natural and particularly useful application. In the following, we elaborate on this scenario and prove Proposition~2 from the main text.

States equilibrated with respect to a thermal reservoir at a temperature $T$ are the thermal states of the form \eqref{eq:thermal}. As in the case of adiabatic processes, we also consider states that are (normalised) projections of \eqref{eq:thermal} onto a subspace (i.e., states where only a subset of the energy eigenstates are populated) as members of the class of ``equilibrium states''. Note that the notion of an equilibrium state depends on the Hamiltonian of the system. 

The possible state transitions of a system connected to a thermal reservoir at a temperature $T$ are modelled by the resource theory of thermal operations introduced in Section~\ref{sec:thermal}. According to Proposition~\ref{prop:HO} from~\cite{Horodecki2013_ThermoMaj}, the ordering of states under these operations is encompassed by means of the order relation of thermo-majorization, as long as the considered states are block-diagonal in the energy eigenbasis.

As equilibrium states are always diagonal in the energy eigenbasis, the relation $\prec_{\mathrm{T}}$ fully captures their inter convertibility. For those non-equilibrium states which are not block diagonal in the energy eigenbasis, the relation $\prec_{\mathrm{T}}$ is not sufficient to express whether two states can be inter converted by thermal operations\footnote{More precisely, $\prec_{\mathrm{T}}$ can only relate such non-equilibrium states to states which are block diagonal in the energy eigenbasis and only as long as the state that is not block diagonal is the preceding element in the order, as was previously mentioned in Section~\ref{sec:resources}. Consult~\cite{Horodecki2013_ThermoMaj} for further clarification.}.
In the following, we therefore restrict our treatment to equilibrium states and to non-equilibrium states that are block diagonal in the energy eigenbasis.
The following proposition ensures the applicability of Theorem~\ref{thm:entropythm} and Proposition~\ref{prop:nonequ}
to quantum states ordered with $\prec_{\mathrm{T}}$.
\begin{prop} \label{prop:thermaxioms}
Consider the order relation of thermo-majorization $\prec_{\mathrm{T}}$. Then, for equilibrium states $\tau$ the six axioms E1 to E6 as well as the Comparison Hypothesis hold, 
whereas all block diagonal non-equilibrium states satisfy axioms N1 and N2. 
\end{prop}

By Definition~\ref{definition:thermom}, block diagonal states $\rho \in \mathcal{S}(\mathcal{H})$ can be represented by their Gibbs-rescaled step functions $f_{\rho}^{\mathrm{T}}$.
For equilibrium states $\tau$ the functions $f^{\mathrm{T}}_{\mathrm{\tau}}$ take the simple form
\begin{equation*}
f^{\mathrm{T}}_{\mathrm{\tau}}(x)= \begin{cases} \frac{1}{Z_\tau}, &  0 \leq x \leq Z_\tau,\\
 0, &  \textrm{otherwise}, 
\end{cases}
\end{equation*}
where $Z_\tau$ is the partition function on the subspace of $\tau$.
We define the composition of two arbitrary states as their tensor product, like before. 
For scaling factors $\lambda \in \mathbbm{N}$ the scaling of an equilibrium state is again assumed to coincide with its composition.
This scaling operation can be formally extended to any scaling factor $\lambda \in \mathbbm{R}_{>0}$ 
on the level of the functions $f^{\mathrm{T}}$.
The Gibbs-rescaled step function of an equilibrium state $\lambda \tau$ has the form 
\begin{equation*}
f^{\mathrm{T}}_{\mathrm{\lambda \tau}}(x)= 
\begin{cases} \left( \frac{1}{Z_\tau} \right)^{\lambda}, &  0 \leq x \leq Z_\tau^{\lambda},\\
	 0, &  \textrm{otherwise}. 
\end{cases}
\end{equation*}

Note that because these functions are flat, we can, similarly to Section~\ref{sec:results}, give a meaning to scaling with non-integer factors $\lambda$ by considering states on a larger Hilbert space.

The Gibbs-rescaled $f^{\mathrm{T}}_{\mathrm{\rho}}$ are normalized and monotonically decreasing like the $f_{\mathrm{\rho}}$.
Furthermore, for equilibrium states $\tau$ the partition function on the subspace where they live, $Z_\tau$, fully determines which states thermo-majorize which others. More specifically, for
two equilibrium states $\tau$ and $\tilde{\tau}$, 
\begin{align}
\begin{split}
\tau \prec_{\mathrm{T}} \tilde{\tau} &\Leftrightarrow \int_{0}^{k}f^{\mathrm{T}}_{\mathrm{\tau}}(x) dx \ \geq \ \int_{0}^{k} f^{\mathrm{T}}_{\mathrm{\tilde{\tau}}}(x) dx, \quad \forall \ k \in \mathbbm{R}_{\mathrm{\geq 0}}\\
&\Leftrightarrow Z_{\mathrm{\tilde{\tau}}} \geq Z_{\mathrm{\tau}}.
\end{split}\label{eq:partition}
\end{align}
Thus for the $f^{\mathrm{T}}_{\mathrm{\rho}}$ the partition function takes the role of the $\operatorname{rank}$ in $f_{\mathrm{\rho}}$.

\begin{proof}[Proof of Proposition~\ref{prop:thermaxioms}]
Substituting the $f_{\mathrm{\rho}}$ with the $f^{\mathrm{T}}_{\mathrm{\rho}}$ and the $\operatorname{rank}$ with the partition function $Z$, we can apply the proof of Lemma~\ref{lemma:major} to prove the axioms E1 to E6 as well as the 
Comparison Hypothesis for the order relation of thermo-majorization $\prec_{\mathrm{T}}$. 
Note that when adapting the proof of E3, the choice of the $\left\{m_{i}\right\}_{i}$ and the $\left\{n_{i}\right\}_{i}$ is not problematic, 
as all energies are positive and hence $0 \leq e^{-\beta E_i} \leq 1$.  
Furthermore, axiom N1 holds as the spectrum of a state $\rho =  \sum_{i,d_{i},d'_{i}} \rho_{i,d_{i},d'_{i}} \left|E_i,~d_{i} \right\rangle \left\langle E_i,~d'_{i} \right| \in \mathcal{S}(\mathcal{H})$ always 
thermo-majorizes $\tau_{1} =  \sum_{i,d_{i}} \frac{e^{-\beta E_i}}{Z} \left|E_i,d_{i} \right\rangle \left\langle E_i, d_{i} \right|$ and is
thermo-majorized by $\tau_{0} = \left|E_1 \right\rangle \left\langle E_1 \right|$. The $d_i$ and $d'_i$ account for the degeneracy of each energy level.
\end{proof}

As all axioms are satisfied by the order relation $\prec_{\mathrm{T}}$, Theorem~\ref{thm:entropythm} implies that there is a unique entropy function 
\begin{equation*}
S_{\mathrm{T}}(\rho)=\left\{ \lambda : ((1-\lambda) \tau_{\mathrm{0}}, \lambda \tau_{\mathrm{1}}) \sim_{\mathrm{T}} \tau \right\},
\end{equation*} 
which is additive under composition by means of the tensor product, extensive under the scaling operation and monotonic under thermal operations. There are furthermore bounds on the monotonic extensions of this function to other block-diagonal states.

\begin{proof}[Proof of Proposition~2 from the main text]
Let $\rho \in \mathcal{S}(\mathcal{H})$ be block diagonal with rescaled eigenvalues $p_{i}^{\mathrm{res}}=\frac{p_{i}}{e^{-\beta E_i}}$ ordered 
as $p_{1}^{\mathrm{res}} \geq p_{2}^{\mathrm{res}} \geq \ldots \geq p_{d}^{\mathrm{res}}$ with $d=\dim(H)$ and let $\tau_{0} \prec\prec_{\mathrm{T}} \tau_{1}$ 
be two equilibrium states.
Let $\lambda$ be such that $((1-\lambda)\tau_{\mathrm{0}}, \lambda \tau_{\mathrm{1}}) \prec_{\mathrm{T}} \rho$.
Then, 
\begin{equation}
 \int_{0}^{k} f^{\mathrm{T}}_{((1-\lambda)\tau_{\mathrm{0}}, \lambda \tau_{\mathrm{1}})}(x)dx \geq \int_{0}^{k} f^{\mathrm{T}}_{\mathrm{\rho}}(x)dx, \quad \forall k \in \mathbbm{R}_{\mathrm{\geq 0}}. 
\label{eq:71} \end{equation}
As $\tau_{\mathrm{0}}$ and $\tau_{\mathrm{1}}$ are equilibrium states, for any
$0 \leq k \leq Z_{\tau_{\mathrm{0}}}^{\mathrm{1-\lambda}}Z_{\tau_{\mathrm{1}}}^{\mathrm{\lambda}}$
\begin{equation*} 
\int_{0}^{k} f^{\mathrm{T}}_{((1-\lambda)\tau_{\mathrm{0}}, \lambda \tau_{\mathrm{1}})}(x)dx = \int_{0}^{k} \left(\frac{1}{Z_{\mathrm{\tau_{\mathrm{0}}}}} \right)^{\mathrm{1-\lambda}} 
\left(\frac{1}{Z_{\mathrm{\tau_{\mathrm{1}}}}} \right)^{\mathrm{\lambda}} dx. 
\end{equation*} 
Therefore, \eqref{eq:71} implies 
\begin{equation}
\left(\frac{1}{Z_{\mathrm{\tau_{\mathrm{0}}}}} \right)^{\mathrm{1-\lambda}} \left(\frac{1}{Z_{\mathrm{\tau_{\mathrm{1}}}}} \right)^{\mathrm{\lambda}} \geq p_{1}^{\mathrm{res}},
\label{eq:conditionZ1} 
\end{equation}
which can be rewritten as
\begin{equation}
a_{\mathrm{T}} \cdot \ln \frac{1}{p_{1}^{\mathrm{res}}}+ b_{\mathrm{T}} \geq \lambda.
\label{eq:ble} \end{equation}
with $a_{\mathrm{T}}=\frac{1}{\ln \frac{Z_{\mathrm{\tau_{1}}}}{Z_{\mathrm{\tau_{0}}}}}$ and 
$b_{\mathrm{T}}=-a_{\mathrm{T}} \cdot \ln Z_{\mathrm{\tau_{0}}}$ depending on the gauge states $\tau_{\mathrm{0}}$ and $\tau_{\mathrm{1}}$.
On the other hand, \eqref{eq:conditionZ1} implies
\begin{align*}
\left(\frac{1}{Z_{\tau_{\mathrm{0}}}} \right)^{\mathrm{1-\lambda}} \left(\frac{1}{Z_{\tau_{\mathrm{1}}}} \right)^{\mathrm{\lambda}} \min \left\{k, \tilde{k} \right\} 
&\geq p^{\mathrm{res}}_{1} \cdot \min \left\{k, \frac{1}{p^{\mathrm{res}}_{1}}\right\} \\
&\geq \int_{0}^{k} f^{\mathrm{T}}_{\mathrm{\rho}}(x)dx , 
\end{align*}
for all $k \in \mathbbm{R}_{\geq 0}$; the second inequality holds as $f^{\mathrm{T}}_{\rho}(x)$ is monotonously decreasing and normalised.
Hence, conditions \eqref{eq:71} and \eqref{eq:conditionZ1} are equivalent.

For a state $\rho$ that is block diagonal in the energy eigenbasis 
\begin{align*}
\begin{split} F_{\mathrm{max}}(\rho) &= -k_{\mathrm{B}}T \ln Z_{\mathrm{\tau}}+k_{\mathrm{B}}TD_{\mathrm{\infty}}(\rho||\tau) \ln(2)\\
&= -k_{\mathrm{B}}T \ln Z_{\mathrm{\tau}}+ k_{\mathrm{B}}T \ln \min \{\lambda: \rho \leq \lambda \tau\}\\
&= k_{\mathrm{B}}T \ln \min \{\mu: \rho \leq Z_{\tau} \mu \tau \}\\
&= k_{\mathrm{B}}T \ln p_{\mathrm{max}}^{\mathrm{res}}
,\end{split}
\end{align*}
where $p_{\mathrm{max}}^{\mathrm{res}}$ is the maximal rescaled eigenvalue of the state $\rho$. 
Thus, taking the supremum over $\lambda$ in \eqref{eq:ble} implies that $\tilde{S}_{\mathrm{T-}}(\rho)=F_{\mathrm{max}}(\rho)$ up 
to affine change of scale.

\vspace{\baselineskip}

For $\tilde{S}_{\mathrm{T+}}$ the proof works similarly.
Let $\rho \in \mathcal{S}(\mathcal{H})$ and let $\tau_{0} \prec\prec_{T} \tau_{1}$ be two equilibrium states. 
Now let $\lambda$ be such that $\rho \prec_{\mathrm{T}} ((1-\lambda)\tau_{\mathrm{0}},\lambda\tau_{\mathrm{1}})$ and thus 
\begin{equation} 
\int_{0}^{k} f^{\mathrm{T}}_{((1-\lambda)\tau_{\mathrm{0}}, \lambda\tau_{\mathrm{1}})}(x)dx
\leq  \int_{0}^{k} f^{\mathrm{T}}_{\mathrm{\rho}}(x)dx, \quad \forall \ k \in \mathbbm{R}_{\mathrm{\geq 0}}. \label{eq:blo}
\end{equation}
First we show by contradiction that 
\begin{equation}
Z_{\tau_{\mathrm{0}}}^{\mathrm{1-\lambda}} Z_{\tau_{\mathrm{1}}}^{\lambda} \geq Z_{\rho}.
\label{eq:contr} \end{equation}
Assume for now that $Z_{\tau_{\mathrm{0}}}^{\mathrm{1-\lambda}} Z_{\tau_{\mathrm{1}}}^{\lambda} < Z_{\rho}$. 
For $\tilde{k} = Z_{\tau_{\mathrm{0}}}^{\mathrm{1-\lambda}} Z_{\tau_{\mathrm{1}}}^{\lambda}$ we therefore find
\begin{equation*}
\int_{0}^{\tilde{k}} f^{\mathrm{T}}_{\mathrm{\rho}}(x)dx < 1.
\end{equation*}
This contradicts \eqref{eq:blo} as 
\begin{equation*}
\int_{0}^{\tilde{k}} f^{\mathrm{T}}_{((1-\lambda)\tau_{\mathrm{0}}, \lambda\tau_{\mathrm{1}})}(x)dx =1.
\end{equation*}
Thus we have $\tilde{k} \geq Z_{\rho}$, which can be rewritten as 
\begin{equation} 
\lambda \geq a_{\mathrm{T}} \cdot \ln Z_{\rho} + b_{\mathrm{T}},\label{eq:bleu}
\end{equation}
with $a_{\mathrm{T}}$ and $b_{\mathrm{T}}$ defined as above.
Moreover, \eqref{eq:contr} implies 
\begin{align*}
\left(\frac{1}{Z_{\tau_{\mathrm{0}}}} \right)^{\mathrm{1-\lambda}} \left(\frac{1}{Z_{\tau_{\mathrm{1}}}} \right)^{\mathrm{\lambda}} 
\min \left\{k, Z_{\tau_{\mathrm{0}}}^{\mathrm{1-\lambda}} Z_{\tau_{\mathrm{1}}}^{\lambda} \right\} 
&\leq \frac{1}{Z_{\rho}} \min \left\{k, Z_{\rho} \right\} \\
&\leq \int_{0}^{k} f^{\mathrm{T}}_{\mathrm{\rho}}(x)dx,
\end{align*}
for all $k \in \mathbbm{R}_{\geq 0}$, i.e., it implies \eqref{eq:blo}; the second inequality holds as $f^{\mathrm{T}}_{\mathrm{\rho}}(x)$ 
is monotonously decreasing and normalized.

Furthermore, we find that
\begin{align*} 
F_{\mathrm{min}}(\rho) &= - k_{\mathrm{B}}T\ln Z_{\mathrm{\tau}}+k_{\mathrm{B}}TD_{\mathrm{0}}(\rho||\tau) \ln(2) \\
&=-k_{\mathrm{B}}T\ln Z_{\mathrm{\tau}} - k_{\mathrm{B}}T\ln \operatorname{tr} \Pi_{\mathrm{\rho}} \tau \\
&=-k_{\mathrm{B}}T\ln\left( Z_{\tau} \operatorname{tr} \Pi_{\mathrm{\rho}} \tau \right)\\
&=-k_{\mathrm{B}}T\ln Z_{\rho},
\end{align*}
where $\Pi_{\mathrm{\rho}}$ denotes the projector onto the support of $\rho$, as before. 
Taking the infimum over $\lambda$ in \eqref{eq:bleu} implies that  $\tilde{S}_{\mathrm{T+}}(\rho)=F_{\mathrm{min}}(\rho)$ up to affine change of scale and concludes the proof for $\tilde{S}_{\mathrm{T+}}$.

For an equilibrium state $\tau$ and for $\lambda= S_{\mathrm{T}}(\tau)$ the relation $((1-\lambda) \tau_{\mathrm{0}}, \lambda \tau_{\mathrm{1}}) \sim_{\mathrm{T}} \tau $ holds. Furthermore, $\ln Z_\tau = \ln \frac{1}{p_\mathrm{max}^{\mathrm{res}}}$ holds for such states. Hence, according to the above considerations, 
\begin{equation*}
\lambda=a_{\mathrm{T}} \cdot \ln Z_{\mathrm{\tau}} + b_{\mathrm{T}}.
\end{equation*}
This concludes the proof.
\end{proof}
The function $S_{\mathrm{T}}(\tau)$ equals the Helmholtz free energy 
\begin{equation*} 
F(\tau) = - k_{\mathrm{B}} T \ln Z_{\tau} 
\end{equation*}
up to an affine change of scale and a factor $-1$.\footnote{
Note that one would require two gauge states with partition functions $Z_{\tau_{\mathrm{0}}}=1$ and $Z_{\tau_{\mathrm{1}}}=e^{-\beta}$ to formally recover $S_{\mathrm{T}}(\tau)=F(\tau)$. 
This choice of parameters, however, would obey $\tau_{\mathrm{1}} \prec \prec_{\mathrm{T}} \tau_{\mathrm{0}}$ and not the required
$\tau_{\mathrm{0}} \prec\prec_{\mathrm{T}} \tau_{\mathrm{1}}$.}
This additional factor reflects that $F$ is not monotonically increasing but rather decreasing under thermal operations.

That the quantities $\tilde{S}_{\mathrm{T-}}$ and $\tilde{S}_{\mathrm{T+}}$ calculated within our axiomatic approach are related to the work of formation $F_{\mathrm{max}}$ as well as to the extractable work $F_{\mathrm{min}}$ from~\cite{Horodecki2013_ThermoMaj}, is perhaps not surprising, since both approaches employ the tools of majorization and thermo-majorization.

\section{Interaction with other types of reservoirs} \label{sec:results3}
Adding other reservoirs to a physical system leads to mathematically equivalent situations, even though the underlying physics differs. Here, we outline the two other scenarios from Table~I of the main text.

\subsection{Interaction with a heat and a particle reservoir}
In addition to a heat bath, considering a particle reservoir is a common practice, especially in statistical physics. 
In the context of resource theories for quantum thermodynamics, such scenarios have been looked at in~\cite{Halpern2014}.

Systems in contact with both a heat and a particle reservoir, have equilibrium states 
\begin{equation}\label{eq:actualequil}
\rho= \sum_{E,N}\frac{e^{ -\beta(E-N\mu)}}{\mathcal{Z}}\left| E,~N \right\rangle \left\langle E,~N \right|
\end{equation}
in the eigenbasis $\left\{\left| E,~N \right\rangle \right\}_{E,~N}$, where $E$ and $N$ denote energy and particle number respectively and $\mathcal{Z}$ is the grand canonical partition function. Like in the case of a heat bath, we also include states for which only a subset of the eigenstates are populated as ``equilibrium states''.~\footnote{For simplicity we assume that there are only particles of one kind.} 

For quantum states $\rho$ that are block diagonal in the energy-particle eigenbasis, the possibility of such processes can be expressed by an order relation $\prec_{\mathrm{N,T}}$, 
which consists again of a rescaling followed by majorization. 
Analogous to the Gibbs-rescaling, we can define a rescaled step function $f_{\mathrm{\rho}}^{\mathrm{N,T}}$. 
\begin{definition}
Consider a block diagonal state in $\mathcal{S}(\mathcal{H})$:
\begin{equation*}
\rho=\sum_{\substack{E,N,n_{E,N},\\  n'_{E,N}}} \rho_{E,N,n_{E,N},n'_{E,N}} \left|E,N,n_{E,N} \right\rangle \left\langle E,N,n'_{E,N} \right|.
\end{equation*}
Its spectrum can be denoted by a step function according to \eqref{eq:stepfunction}. Then, the
 \textbf{N,T-rescaled step function} of $\rho$ is 
\begin{widetext}
\begin{equation*} 
f_{\mathrm{\rho}}^{\mathrm{N,T}}(x)=\begin{cases} \frac{p_{i}}{e^{-\beta (E_i-N_i\mu)}}, &  \sum_ {k=1}^{i-1} e^{-\beta(E_k-N_k\mu)} \leq x \leq \sum_ {k=1}^{i} e^{-\beta(E_k-N_k \mu)},\\
0, &  \textrm{otherwise.} 
\end{cases}
\end{equation*}
\end{widetext}
\end{definition}
For an equilibrium state $\rho$ this is
\begin{equation*}
f^{\mathrm{N,T}}_{\mathrm{\rho}}(x)= \begin{cases} \frac{1}{\mathcal{Z}}, &  0 \leq x \leq \mathcal{Z},\\
0, &  \textrm{otherwise,} 
\end{cases}
\end{equation*}
where $\mathcal{Z}$ is the grand canonical partition function of the occupied states.
\begin{definition}
Let $\rho$, $\sigma \in \mathcal{S}(\mathcal{H})$ be two states which are block diagonal in the basis 
$\left\{ \left| E,~N \right\rangle \right\}_{E,~N}$.
The order of \textbf{N,T-majorization $\prec_{\mathrm{N,T}}$} is defined as
\begin{equation*}
\rho \prec_{\mathrm{N, T}} \sigma \, \Leftrightarrow \, \int_{0}^{k} f^{\mathrm{N,T}}_{\mathrm{\rho}}(x) dx \geq \int_{0}^{k} 
f^{\mathrm{N,T}}_{\mathrm{\sigma}}(x) dx, \: \forall k \in \mathbbm{R}_{\mathrm{\geq 0}}.
\end{equation*}
\end{definition}
This is analogous to Definition \ref{definition:thermom} but with an adapted rescaling operation.
In analogy to Section~\ref{sec:results2}, the operations corresponding to this order relation $\prec_{\mathrm{N,T}}$ are:
\begin{itemize}
\item composition with an ancillary systems in a state of the form \eqref{eq:actualequil};
\item unitary transformation conserving the total energy and particle number of system and ancilla;
\item removal of any subsystem.
\end{itemize}
For equilibrium states -- which are always diagonal in the basis $\left\{\left| E,~N \right\rangle \right\}_{E,~N}$-- this treatment suffices, whereas for those non-equilibrium states which contain coherent superpositions of such eigenstates it does not apply.

As for $\prec_{\mathrm{T}}$, the relation $\prec_{\mathrm{N,T}}$ fulfills Lieb and Yngvason's axioms. According to Theorem \ref{thm:entropythm} there is thus a unique potential
\begin{equation*} S_{\mathrm{N,T}}(\rho)= \left\{ \lambda : ((1-\lambda) \rho_{\mathrm{0}}, \lambda \rho_{\mathrm{1}}) \sim_{\mathrm{N,T}} \rho \right\} 
\end{equation*}
for equilibrium states.
\begin{prop}
For an equilibrium state $\rho$ 
the unique function $S_{\mathrm{N,T}}$ is $$S_{\mathrm{N,T}}(\rho)= a_{\mathrm{N,T}} \cdot \ln \mathcal{Z} + b_{\mathrm{N,T}}$$ with $a_{\mathrm{N,T}}>0$.
\end{prop}
\begin{proof} 
Let $\rho \in \mathcal{S}(\mathcal{H})$ be an equilibrium state and let the equilibrium states $\rho_{\mathrm{0}} \prec\prec_{\mathrm{N,T}} \rho_{\mathrm{1}} \in \mathcal{S}(\mathcal{H})$ define a gauge. 
Then for $\lambda=S_{\mathrm{N,T}}(\rho)$, the relation 
$((1-\lambda) \rho_{\mathrm{0}}, \lambda \rho_{\mathrm{1}}) \sim_{N,T} \rho$ holds and is equivalent to
\begin{equation*}
 \int_{0}^{k} f^{\mathrm{N,T}}_{ \mathrm{((1-\lambda) \rho_{0}, \lambda \rho_{1})}}(x) dx = \int_{0}^{k} f^{\mathrm{N,T}}_{\mathrm{\rho}}(x) dx, \
 \forall \ k \in \mathbbm{R}_{\mathrm{\geq 0}}.
\end{equation*}
As for equilibrium states $\rho \prec_{\mathrm{N,T}} \tilde{\rho} \Leftrightarrow \mathcal{Z}_{\tilde{\rho}} \geq \mathcal{Z}_{\rho}$ this is equivalent to
\begin{equation*} 
\mathcal{Z}_{\mathrm{\rho_{0}}}^{1-\lambda} \mathcal{Z}_{\mathrm{\rho_{1}}}^{\lambda} = \mathcal{Z}_{\mathrm{\rho}}
\end{equation*}
and can be written as
\begin{align*} \begin{split}
\lambda &= \frac{1}{\ln \frac{\mathcal{Z}_{\mathrm{\rho_{1}}}}{\mathcal{Z}_{\mathrm{\rho_{0}}}}} \ln \mathcal{Z}_{\mathrm{\rho}} - \frac{1}{\ln \frac{\mathcal{Z}_{\mathrm{\rho_{1}}}}{\mathcal{Z}_{\mathrm{\rho_{0}}}}}\ln \mathcal{Z}_{\mathrm{\rho_{0}}}\\
&= a_{\mathrm{N,T}} \cdot \ln \mathcal{Z}_{\mathrm{\rho}} + b_{\mathrm{N,T}},
\end{split}\end{align*}
where $a_{\mathrm{N,T}}=\frac{1}{\ln \frac{\mathcal{Z}_{\mathrm{\rho_{1}}}}{\mathcal{Z}_{\mathrm{\rho_{0}}}}}$ and  
$b_{\mathrm{N,T}}=-a_{\mathrm{N,T}} \cdot \ln \mathcal{Z}_{\mathrm{\rho_{0}}}$.
\end{proof}

The ``entropy function'', i.e., the potential for a system in contact with a heat bath and a particle reservoir is related to the grand potential $\Omega$, since
\begin{equation*} 
\Omega = - k_{\mathrm{B}} T \ln \mathcal{Z}.
\end{equation*}

For non-equilibrium states, the bounding functions $\tilde{S}_{\mathrm{N,T -}}$ and $\tilde{S}_{\mathrm{N,T+}}$ can be calculated analogously to the scenario 
including only a heat bath. They define bounds $\Omega_{\mathrm{max}}$ and $\Omega_{\mathrm{min}}$ on the potential $\Omega$ for non-equilibrium states. 
These are operationally related to the formation and destruction of a state in this scenario, again analogously to the case of a heat bath.

\subsection{Interaction with an angular momentum reservoir} \label{sec:angular}
In the study of Landauer's principle~\cite{Landauer1961_5392446Erasure}, the question as to whether energy should obtain a special 
role among the conserved quantities or 
whether processes such as erasure could also be realized at an angular momentum instead of an energy cost was raised in Refs.~\cite{Barnett2013_beyond,Vaccaro2011} and 
answered in the affirmative. 
We observe here that systems connected to an angular momentum reservoir can be included in the axiomatic framework.

As the consideration of an angular momentum reservoir is not common practice, we first sketch the concrete model of a spin reservoir from Refs.~\cite{Barnett2013_beyond,Vaccaro2011}. 
The reservoir consists of N mobile spin-$\frac{1}{2}$ particles, for which the possible spin states are denoted by $\left| 0 \right\rangle$ and $\left| 1 \right\rangle$.  
The spin states are assumed to be degenerate in energy and thus decoupled from the spatial degrees of freedom, which are in equilibrium with a heat bath. 
The equilibrium probability for the reservoir to be in a particular state with $n$ particles in state $\left| 1 \right\rangle$ and $N-n$ particles in state $\left| 0 \right\rangle$ is 
\begin{equation*} 
p_{\mathrm{n}}=\frac{e^{-n \hbar \gamma}}{(1+e^{- \hbar \gamma})^{N}}, 
\end{equation*}
where the parameter $\gamma$ is the analogue of the inverse temperature $\beta$ encountered in the context of a heat bath. 
As for each value $n$ there are $N \choose n$ such reservoir states, the normalization is given as $Z_{\mathrm{J}}^{\mathrm{res}}=(1+e^{- \hbar \gamma})^{N}$ and
has the form of a partition function for the angular momentum reservoir.
This construction allows us to consider an angular momentum reservoir of arbitrary size $N$. 
In the following, we consider a reservoir in the limit $N \rightarrow \infty$.

The state of a system in contact with a spin angular momentum reservoir can be described by a density operator. 
Again, we restrict our considerations to states 
that are block diagonal, this time in the eigenbasis of the z-component of the spin operator, denoted $\{\left| J \right\rangle\}_{J}$.
To ensure that energy does not affect our considerations, we assume all spin-levels $\left| J \right\rangle$ to be energetically degenerate.
A system in equilibrium with the reservoir is described by a density operator of the form
\begin{equation}\label{eq:jequ}
\rho=\sum_{J} \frac{e^{-J \hbar \gamma}}{Z_{\mathrm{J}}} \left| J \right\rangle \left\langle J \right|
\end{equation}
with partition function $Z_{\mathrm{J}}=\sum_{i} e^{- J \hbar \gamma}$. (Note that we include states for which only a subset of the eigenstates are populated, as long as the projection onto the corresponding subspace has this form.)
The spectra of block diagonal states can be represented in terms of rescaled step functions.

\begin{definition} 
Let 
\begin{equation*}
\rho=\sum_{J,n_J,n'_J} \rho_{J,n_J,n'_J} \left| J,~n_J \right\rangle \left\langle J,~n'_J \right|
\end{equation*}
be a density matrix block diagonal in the eigenbasis of the z-component of the spin operator. Represent its spectrum according to~\eqref{eq:stepfunction}. 
Then the \textbf{J-rescaled step function} of $\rho$ is defined as
\begin{equation*}
f_{\mathrm{\rho}}^{J}(x) \!=\!
\begin{cases}\frac{p_{\mathrm{i}}}{ e^{-J(i)} \hbar \gamma}, & \! \sum_ {k=1}^{i-1}  e^{-J{(k)} \hbar \gamma} \! \leq \! x \! \leq \! \sum_ {k=1}^{i}  e^{-J{(k)} \hbar \gamma},\\
0, & \! \textrm{otherwise.} 
\end{cases}
\end{equation*}
\end{definition}
For an equilibrium state $\rho$ this simplifies to 
\begin{equation*} 
f_{\rho}^{J}(x)=
\begin{cases}\frac{1}{Z_{\mathrm{J}}}, &  0 \leq x \leq Z_{\mathrm{J}},\\
0, &  \textrm{otherwise.} 
\end{cases} 
\end{equation*}

We describe the processes on a system in contact with an angular momentum reservoir as:
\begin{itemize} 
\item composition with an ancillary system in an equilibrium state of the form \eqref{eq:jequ};
\item unitary transformation of system and ancilla conserving angular momentum;
\item removal of any subsystem.
\end{itemize}
These lead to an ordering of the states by means of an order relation $\prec_{\mathrm{J}}$.

\begin{definition} Let $\rho$, $\sigma \in \mathcal{S}(\mathcal{H})$ be two states which are block diagonal in the basis $\left\{ \left| J \right\rangle \right\}_{J}$.
The relation of \textbf{J-majorization $\prec_{\mathrm{J}}$} is defined as
\begin{equation*}
\rho \prec_{\mathrm{J}} \sigma \quad \Leftrightarrow \quad \int_{0}^{k} f^{\mathrm{J}}_{\mathrm{\rho}}(x) dx \geq \int_{0}^{k} 
f^{\mathrm{J}}_{\mathrm{\sigma}}(x) dx, \quad \forall k \in \mathbbm{R}_{\mathrm{\geq 0}}.
\end{equation*}
\end{definition}

As in the case of a heat bath, the rescaled majorization relation $\prec_{\mathrm{J}}$ fulfills axioms E1 to E6 for equilibrium states and N1 and N2 for non-equilibrium states. 
This gives rise to a unique potential, 
\begin{equation*} 
S_{\mathrm{J}} \propto \ln Z_{\mathrm{J}},
\end{equation*}
for equilibrium states.
For non-equilibrium states, the two quantities $\tilde{S}_{\mathrm{J-}}$ and $\tilde{S}_{\mathrm{J+}}$ provide necessary conditions as well as a sufficient condition for state transformations under the above operations.

We have thus found an angular momentum based resource theory corresponding to the order relation $\prec_{\mathrm{J}}$, 
for which $S_{\mathrm{J}}$, $\tilde{S}_{\mathrm{J-}}$ and $\tilde{S}_{\mathrm{J+}}$ are monotones. 
In agreement with~\cite{Barnett2013_beyond,Vaccaro2011},
energy can be substituted with angular momentum in our resource theoretic picture.
------------------------------------------------------------------------------

\end{document}